\documentclass[10pt, a4paper,reqno]{amsart}

\usepackage{color} \definecolor{bleu_sombre}{rgb}{0,0,0.6}  \definecolor{rouge_sombre}{rgb}{0.8,0,0}\definecolor{vert_sombre}{rgb}{0,0.6,0}
\usepackage[plainpages=false,colorlinks,linkcolor=bleu_sombre,citecolor=rouge_sombre,urlcolor=vert_sombre,breaklinks]{hyperref}

\usepackage[english]{babel}
\usepackage{amsmath,amssymb,amsthm,graphicx,amsfonts,url,color,enumerate,dsfont,mathrsfs}

\theoremstyle{plain}
\newtheorem{theorem}{{Theorem}}[section] 
\newtheorem*{theorem*}{{Theorem}}
\newtheorem{proposition}[theorem]{Proposition}
\newtheorem*{proposition*}{Proposition}
\newtheorem{corollary}[theorem]{Corollary}
\newtheorem*{corollary*}{Corollary}
\newtheorem{lemma}[theorem]{Lemma}
\newtheorem*{lemma*}{Lemma}

\theoremstyle{definition}

\newtheorem*{definition*}{Definition}

\theoremstyle{remark}
\newtheorem{remark}[theorem]{Remark}

\makeatletter

\@addtoreset{equation}{section}  
\makeatother

\newcommand{\R}{\mathbb{R}}		\newcommand{\C}{\mathbb{C}}
\newcommand{\N}{\mathbb{N}}		

\renewcommand{\a}{\alpha}\renewcommand{\b}{\beta}\newcommand{\g}{\gamma}\newcommand{\G}{\Gamma}\renewcommand{\d}{\delta}\newcommand{\D}{\Delta}\newcommand{\e}{\varepsilon}\newcommand{\z}{\zeta} \newcommand{\y}{\eta}\renewcommand{\th}{\theta}\newcommand{\Th}{\Theta}\renewcommand{\k}{\kappa}\renewcommand{\l}{\lambda}\newcommand{\m}{\mu}\newcommand{\n}{\nu}\newcommand{\x}{\xi}\newcommand{\s}{\sigma}\renewcommand{\t}{\tau}\newcommand{\f}{\varphi}\newcommand{\vf}{\phi}\newcommand{\h}{\chi}\newcommand{\p}{\psi}\renewcommand{\o}{\omega}\renewcommand{\O}{\Omega}

\newcommand{\Ac}{{\mathcal A}}\newcommand{\Dc}{{\mathcal D}}\newcommand{\Kc}{{\mathcal K}}\newcommand{\Lc}{{\mathcal L}}\newcommand{\Pc}{{\mathcal P}}\newcommand{\Rc}{{\mathcal R}}\newcommand{\Uc}{{\mathcal U}}\newcommand{\Vc}{{\mathcal V}}\newcommand{\Zc}{{\mathcal Z}}

\newcommand{\eqv}{\Longleftrightarrow}               
\renewcommand{\leq}{\leqslant}	\renewcommand{\geq}{\geqslant}
\renewcommand{\bar}[1]{\overline{#1}}
\newcommand{\inv}{^{-1}}
\newcommand {\limt}[2]{\xrightarrow[#1 \to #2]{}}
\newcommand{\abs}[1]{\left\vert #1\right\vert}        
\newcommand{\nr}[1]{\left\Vert #1\right\Vert}         
\newcommand{\innp}[2]{\left< #1 , #2 \right>}         
\newcommand{\pppg}[1] {\left< #1 \right>} 	
\newcommand{\1}{\mathds 1}
\newcommand{\st}{\,:\,}	
\newcommand{\Dom}{\Dc}			
\newcommand{\Opw}{{\mathop{\rm{Op}}}_h^w}
\newcommand{\Opwm}{{\mathop{\rm{Op}}}_{h_m}^w}
\newcommand{\singl}[1]{\left\{ #1 \right\}}		
\newcommand{\seq}[2]{\left({#1}_{#2}\right)_{#2 \in\N}} 
\DeclareMathOperator{\Ran}{Ran}
\newcommand{\restr}[2]{\left.#1\right|_{#2}}         
\renewcommand{\Re}{\mathop{\rm{Re}}\nolimits}        
\renewcommand{\Im}{\mathop{\rm{Im}}\nolimits}        
\DeclareMathOperator{\supp}{supp}                    

\newcommand{\stepp}{\noindent {\bf $\bullet$}\quad }
\newcommand{\detail}[1]
{
}

\begin{document}

\newcommand{\EE}{\mathscr E}\newcommand{\EED}{\mathscr E_D}\newcommand{\EEN}{\mathscr E_N}
\newcommand{\HH}{\mathcal H}\newcommand{\HHD}{\mathcal H_D}\newcommand{\HHN}{\mathcal H_N}\newcommand{\HHRd}{\mathcal H_d}
\newcommand{\AcN}{\Ac_N}\newcommand{\AcD}{\Ac_D}\newcommand{\AcRd}{\Ac_{\R^d}}\newcommand{\AcKG}{\Ac_{KG}}
\newcommand{\RNz}{R_N(z)}\newcommand{\RNt}{R_N(\t)}\newcommand{\RN}{R_N}
\newcommand{\RDz}{R_D(z)}\newcommand{\RDt}{R_D(\t)}
\newcommand{\HuO}{H^1(\O)}\newcommand{\HuOp}{H^1(\O)'}
\newcommand{\Huo}{H^1_0(\O)}\newcommand{\Hinv}{H^{-1}(\O)}
\newcommand{\TN}{T_N}\newcommand{\TD}{T_D}
\newcommand{\Rcz}{{(-\DN -iz)\inv}}
\newcommand{\Rchz}{(-\Dz  -i \hat z)\inv}
\newcommand{\hRNz}{\hat R_N(z)}
\newcommand{\RcC}{\Rc_{\textrm{Heat}}(z)}
\newcommand{\Heat}{\textrm{Heat}}
\newcommand{\DN}{\D_N}\newcommand{\DDir}{\D_D}\newcommand{\Dz}{\D_z}\newcommand{\Dx}{\D_x}
\newcommand{\EEO}{\EE^0}
\newcommand{\EEbot}{\EE_0^\bot}
\newcommand{\Acbot}{\Ac_0^\bot}

\title{Energy decay in a wave guide with dissipation at infinity}

\author{Mohamed Malloug}
\address[Mohamed Malloug]{\'Ecole Sup\'erieure des Sciences et de la Technologie de Hammam Sousse, Universit\'e de Sousse, Rue Lamine Abassi, 4011 H. Sousse, Tunisia.}
\email{mallougm70@gmail.com}

\author{Julien Royer}
\address[Julien Royer]{Institut de Math\'ematiques de Toulouse, Universit\'e Toulouse 3, 118 route de Narbonne, F-31062 Toulouse cedex 9, France.}
\email{julien.royer@math.univ-toulouse.fr}


\begin{abstract}
We prove local and global energy decay for the wave equation in a wave guide with damping at infinity. More precisely, the absorption index is assumed to converge slowly to a positive constant, and we obtain the diffusive phenomenon typical for the contribution of low frequencies when the damping is effective at infinity. On the other hand, the usual Geometric Control Condition is not necessarily satisfied so we may have a loss of regularity for the contribution of high frequencies. Since our results are new even in the Euclidean space, we also state a similar result in this case.
\end{abstract}

\maketitle

\section{Introduction and statement of the main results}

In this paper we study the wave equation with stabilisation at infinity, either in the usual Euclidean space or in a wave guide. We state our main results for the case of a wave guide, which was the original motivation of this paper. However some of our estimates are not known in the Euclidean space, so we will also give the analogous statements in this context.\\

Let $d,n \in \N^*$, and let $\o$ be a bounded, open, smooth and connected subset of $\R^n$. We denote by $\O$ the straight wave guide $\R^d \times \o \subset \R^{d+n}$. The main examples which we have in mind are the tube in $\R^3$ ($d = 1$, $\o \subset \R^2$), a layer in $\R^3$ ($d = 2$ and $\o$ is a bounded interval of $\R$) or a strip of $\R^2$ ($d = 1$, $\o \subset \R$). Everywhere in the paper we denote by $(x,y)$ a general point in $\O$, with $x \in \R^d$ and $y \in \o$.\\

Given $u_0 \in H^1(\O)$ and $u_1 \in L^2(\O)$, we first consider on $\O$ the dissipative wave equation with Neumann boundary condition
\begin{equation} \label{wave-neumann} 
\begin{cases}
\partial_t^2 u  -\D u + a \partial_t u = 0  & \text{on  }  \R_+ \times \O, \\
\partial_\n u = 0 & \text{on } \R_+ \times \partial \O,\\
\restr{(u , \partial_t u )}{t = 0} = (u_0, u_1) &  \text{on } \O.
\end{cases}
\end{equation}
The similar problem with Dirichlet boundary condition and the damped Klein-Gordon equation will be discussed below.\\

The function $a$ is the absorption index. It is bounded, takes non negative values and goes to 1 at infinity. More precisely, we assume that there exists $\rho > 0$ such that for $\b \in \N^d$ with $\abs \b \leq \frac d 2 + 1$ we have
\begin{equation} \label{hyp-amort-inf}
\forall x \in \R^d, \forall y \in \o, \quad \abs{\partial_x^\b \big(a(x,y)-1\big)} \leq C_\b \pppg x^{-\rho-\abs \b}.
\end{equation}

If $u$ is a solution of \eqref{wave-neumann}, then its energy at time $t$ is defined by 
\begin{equation} \label{def-E}
E(t) =  \int_\O  \abs{\nabla u (t)}^2  + \int_\O \abs{\partial_t u (t)}^2.
\end{equation}
We can check that
\[
E(t_2) - E(t_1) = -2 \int_{t_1}^{t_2} \int_{\O} a \abs {\partial_t u(t)}^2 \, dt,
\]
so the energy is a non-increasing function of time and the decay is due to the loss in the region where $a > 0$. Our purpose in this paper is to say more about this decay. It is also an important question to understand the decay of the local energy 
\[
E_R(t) = \int_{\abs x \leq R}  \abs{\nabla u (t)}^2  + \int_{\abs x \leq R}  \abs{\partial_t u (t)}^2
\]
for any $R > 0$.\\

On a compact manifold, it is now well known since the stabilisation results of \cite{raucht74} (for a manifold without boundary) and \cite{bardoslr92} (for stabilisation at the boundary) that we have uniform (therefore exponential) decay for the energy of the damped wave equation under the so-called Geometric Control Condition. Roughly speaking, the assumption is that any ray of light (trajectory for the flow of the underlying classical problem) should meet the damping region.

For the undamped wave equation on unbounded domains, we have uniform decay of the energy on any compact under the similar non-trapping condition, which says that all the classical trajectories should escape to infinity. Notice that since the total energy is conserved in this case, it is equivalent to say that the energy on any compact goes to 0 or that all the energy escapes to infinity.

This is in particular the case for the free wave equation on $\R^d$ by the Huygens principle. 
For compact perturbations of this model case and under the non-trapping assumption, the energy on any compact decays exponentially in odd dimensions and at rate $t^{-2d}$ if the dimension $d$ is even. See \cite{morawetzrs77} and \cite{melrose79}. See also \cite{ralston69} for the necessity of the geometric assumption. In \cite{bonyh12} and \cite{bouclet11} the problem is given by long-range perturbation of the free wave equation. The local energy (defined with a polynomially decaying weight) decays at rate $O(t^{-2d + \e})$ for any $\e > 0$ in this case.\\

In this paper we are interested in the (local and global) energy decay for the damped wave equation on an unbounded domain. The local energy decay in an exterior domain (with stabilisation at the boundary or in the interior of the domain) has been studied in \cite{alouik02,khenissi03}. We also mention \cite{boucletr14,royer-dld-energy-space} for a non-compact perturbation of the free model. The decay rates are the same as for the corresponding undamped problems, but the non-trapping condition can be replaced by the assumption that all the classical trajectories go through the region where the damping is effective or escape to infinity. 

If all the classical trajectories go through the damping region, and not only the bounded ones, we can obtain decay estimates for the total energy. We mention for instance \cite{AlouiIbKh15} for the wave equation in an exterior domain with damping at infinity and \cite{BurqJo} for the damped Klein-Gordon equation in $\R^d$.\\

In our setting the Geometric Control Condition is not necessarily satisfied. For instance, if there exists $x_0 \in \R^d$ such that $a(x_0,y) = 0$ for all $y \in \o$ then any trajectory staying in $\singl{x_0} \times \o$ will neither see the damping nor escape to infinity. This means that some high frequency solutions may stay in a bounded region without going through the damping region for a very long time, so we cannot have any uniform decay for the local energy (or, a fortiori, for the global energy). However, if we allow some loss of regularity we may have some energy decay. Such results were given in \cite{lebeau96, lebeaur97} for the damped wave equation on a compact domain and in \cite{burq98} for the undamped wave equation equation in an exterior domain. These papers give the minimal decay without any geometric assumption. There are also settings for which G.C.C. fails to hold even if it is satisfied by ``most of the classical trajectories''. This is the case here, 
since all the 
rays which have a non-zero velocity in the $x$ directions escape to infinity. Moreover, outside some bounded subset all the trajectories meet the damping region. We refer for instance to \cite{schenck11,anantharamanl,leautaudl} for a partially damped wave equation, and to \cite{nonnenmacherz09} for the local energy decay with trapped trajectories (more precisely, for the corresponding high frequency resolvent estimates). We have mentioned \cite{BurqJo} above. In this paper the damped Klein-Gordon without G.C.C. is also considered (see also \cite{Wunsch} for a periodic damping). Here the geometry of our trapped and undamped classical trajectories is close to the setting of \cite{Burq-Hi-07} where the wave equation on partially rectangular domains with damping on both ends is considered (see also \cite{Nishiyama-09}).\\

The geometry of undamped and/or trapped trajectories governs the behavior of the contribution of high frequencies. However, in unbounded domains, the general rate of decay is also limited by the contribution of low frequencies. We know that with a short range damping the rate of decay for the local energy is the same as in the undamped case. We also know that with a stronger damping, the local energy decay can be slower than without damping (even if, of course, the global energy decays faster). More precisely, when the damping is effective at infinity, the contribution of low frequencies tends to behave like a solution of a heat equation. For the wave equation with constant dissipation on $\R^d$
\begin{equation} \label{eq-model}
\partial_t^2 u - \D u + \partial_t u = 0,
\end{equation}
this can be understood as follows: for the contribution $u$ of low frequencies the term $\partial_t^2 u$ becomes small compared to $\partial_t u$ and the solution behaves like a solution of the diffusive equation 
\begin{equation} \label{eq-model-heat}
- \D u + \partial_t u = 0.
\end{equation}
For the energy decay of \eqref{eq-model} we refer to \cite{Matsumura76}. For the comparison with the solution of \eqref{eq-model-heat} we refer to \cite{nishihara03,marcatin03, hosonoo04, narazaki04}. An asymptotic expansion of the solution in a periodic setting is given in \cite{OrivePaZu01}. For results on an exterior domain we refer to \cite{Ikehata02} (with constant damping) and \cite{AlouiIbKh15} (the damping index is constant outside a compact and satisfies G.C.C.). When the absorption index $a$ decays slowly ($a(x) \sim \pppg x^{-\rho}$ with $\rho \in ]0,1]$), we have some global energy decay (see for instance \cite{TodorovaYo09} when $\rho < 1$ and \cite{IkehataToYo13} when $\rho = 1$) and we recover the diffusion phenomenon (see \cite{Wakasugi14} when $\rho < 1$). As already mentioned, we recover for the short range case ($\rho > 1$) the same kind of results as in the undamped case (see \cite{Mochizuki76,boucletr14,royer-dld-energy-space}). Finally, results on an abstract setting can be found in \
cite {Chill-Ha-04, Radu-To-Yo-11, nishiyama,Radu-To-Yo-16}.\\

Here we consider the damped wave equation on a wave guide, which is neither compact nor close to the Euclidean space at infinity in any usual sense. Closely related problems have been considered in \cite{royer-diss-schrodinger-guide} (about the dissipative Schr\"odinger equation) and \cite{royer-diss-wave-guide} (about the wave equation with constant dissipation at the boundary).\\

Before stating our results, we introduce the usual notation for the wave equation in the energy space.
We define $\EEN$ as the Hilbert completion of $C_0^\infty(\bar \O) \times C_0^\infty(\bar \O)$ for the norm 
\[
\nr{(u,v)}_{\EEN}^2 = \nr{\nabla u}_{L^2(\O)}^2 + \nr{v}_{L^2(\O)}^2
\]
($C_0^\infty(\bar \O)$ is the set of restrictions to $\bar \O$ of functions in $C_0^\infty(\R^{d+n})$). Given $\d \in \R$, we denote by $\EEN^\d$ the weighted energy space defined as the Hilbert completion of $C_0^\infty(\bar \O) \times C_0^\infty(\bar \O)$ for the norm 
\[
\nr{(u,v)}_{\EEN^\d}^2 = \nr{\pppg x^\d \nabla u}_{L^2(\O)}^2 + \nr{\pppg x^\d v}_{L^2(\O)}^2,
\]
where $\pppg x$ stands for $\big( 1 + \abs x^2 \big)^{\frac 12}$. We also denote by $\HHN^\d$ the Hilbert completion of $C_0^\infty(\bar \O) \times C_0^\infty(\bar \O)$ for the norm 
\[
\nr{(u,v)}_{\HHN^\d}^2 = \nr{\pppg x^\d u}_{L^2(\O)}^2 + \nr{\pppg x^\d \nabla u}_{L^2(\O)}^2 + \nr{\pppg x^\d v}_{L^2(\O)}^2.
\]
We write $\HHN$ instead of $\HHN^0$.\\

We consider on $\EEN$ the operator $\AcN$ defined by 
\begin{equation} \label{def-Ac}
\AcN =
\begin{pmatrix} 0 & I \\ -\D & -i a \end{pmatrix}
\end{equation}
on the domain
\begin{equation} \label{dom-Ac-N}
\Dom(\AcN) = \singl{(u,v) \in \EEN \st \AcN (u,v) \in \EEN \text{ and } \partial_\n u = 0 \text{ on } \partial \O}.
\end{equation}
Now let $U_0 = (u_0 ,iu_1) \in \Dom(\AcN)$. It is standard that $u$ is a solution of \eqref{wave-neumann} if and only if $U : t \mapsto (u(t),i\partial_t u(t))$ is solution of 
\begin{equation} \label{wave-Ac}
\begin{cases}
\partial_t U(t) + i \AcN U(t) = 0, \quad t \geq 0,\\
U(0) = U_0.
\end{cases}
\end{equation}
Since the operator $\AcN$ is maximal dissipative (see Proposition \ref{prop-AcN-diss} below), we know from the Hille-Yosida theorem that $-i\AcN$ generates a contractions semigroup, so that the problem \eqref{wave-Ac} has a unique solution $U : t \mapsto e^{-it\AcN} U_0 \in C^0 (\R_+,\Dom(\AcN)) \cap C^1(\R_+,\EEN)$.\\

Our purpose is to prove local and global energy decay for the solution of \eqref{wave-neumann}. The main result of this paper is the following:

\begin{theorem}[Energy decay] \label{th-energy-decay}
Let $k \in \N^*$, $s_1,s_2 \in \big[0,\frac d 2]$, $\k > 1$, $s \in [0,\min(d,\rho)[$ with $s \leq 1$, $\d_1 \geq  \k s_1 + s$ and $\d_2 \geq  \k s_2 + s$. Then there exists $C \geq 0$ such that for $t \geq 1$ and $U_0 \in \Dom(\AcN^k)$ with $(\AcN-i)^k \in \HHN^{\d_2}$ we have 
\[
\nr{e^{-it\AcN} U_0}_{\EEN^{-\d_1}} \leq C \left( t^{-\frac 12 (1 + s_1 + s_2 + s)} + \frac {\ln(t)^{k/2+1}} {t^{k/2}} \right) \nr{(\AcN-i)^k U_0}_{\HHN^{\d_2}}.
\]
More precisely, if we write $U_0 = (u_0,iu_1)$ and $e^{-it\AcN} U_0 = (u(t),i\partial_t u(t))$ where $u$ is the solution of \eqref{wave-neumann} then we have
\[
\nr{\nabla u(t)}_{L^{2,-\d_1}(\O)} \leq C \left( t^{-\frac 12 (1 + s_1 + s_2 + s)} + \frac {\ln(t)^{k/2+1}} {t^{k/2}} \right) \nr{(\AcN-i)^k U_0}_{\HHN^{\d_2}}
\]
and 
\[
\nr{\partial_t u(t)}_{L^{2,-\d_1}(\O)} \leq C \left( t^{-\frac 12 (2 + s_1 + s_2 )} + \frac {\ln(t)^{k/2+1}} {t^{k/2}} \right) \nr{(\AcN-i)^k U_0}_{\HHN^{\d_2}}.
\]
\end{theorem}

In this result we obtain a polynomial rate of decay. We will see in Theorem \ref{th-chaleur} below that this rate of decay is sharp in general.

The term ${\ln(t)^{k/2+1}} / {t^{k/2}}$ is due to the contribution of high frequencies and depends on the regularity of $U_0$. Under G.C.C., it could be replaced by an exponentially decaying term and the estimate would be uniform. On the other hand the damping is effective at infinity, so for the contribution of high frequencies the energy which escapes to infinity is dissipated. Thus it is equivalent to look at the local or global energy decay and the weights do not play any role.

The first term in the brackets describes the decay for the contribution of low frequencies. It depends on the weights (if $\d_j > 0$ then we can choose $s_j > 0$, which improves the decay). If we want to estimate the global energy without assumption of localization for the initial data, we have to take $s_1 = s_2 = s = 0$ in the theorem. This gives the following estimates:

\begin{corollary}[Uniform global energy decay] \label{cor-global-decay}
There exists $c \geq 0$ such that for $t \geq 1$ and $U_0 = (u_0,iu_1) \in \Dom(\AcN^3)$ we have 
\[
\nr{\nabla u(t)}_{L^2(\O)} \leq \frac c {\sqrt t} \nr{(\AcN-i)^3 U_0}_{\HHN} \quad \and \quad \nr{\partial_t u(t)}_{L^2(\O)} \leq \frac c {t} \nr{(\AcN-i)^3 U_0}_{\HHN},
\]
where $u$ is the solution of \eqref{wave-neumann}.
\end{corollary}

In this result we chose the regularity assumption to ensure that the decay is limited by the contribution of low frequencies. Again, under G.C.C. we obtain the same estimates without loss of regularity.

In \cite{Matsumura76}, \cite{TodorovaYo09} or \cite{AlouiIbKh15}, where global energy decay is studied, the initial data is localized (compactly supported, or at least in $L^2 \cap L^q$ for some $q \in [1,2[$). Theorem \ref{th-energy-decay} also contains this kind of result if we take $\d_1 = 0$ and $\d_2 > 0$. For instance, for compactly supported initial data we can take $s_2 = \frac d 2$ and we obtain the following estimates (as before we take an initial data regular enough to avoid problems with the contribution of high frequencies).

\begin{corollary}[Global energy decay for localized initial data]
Let $K$ be a compact subset of $\bar \O$. Let $k>\frac d 2 + 2$. Then there exists $c \geq 0$ such that for $t \geq 1$ and $U_0 = (u_0,iu_1) \in \Dom(\AcN^k)$ supported in $K \times K$ we have 
\[
\nr{\nabla u(t)}_{L^2(\O)} \leq \frac c {t^{\frac d 4 + \frac 1 2}} \nr{(\AcN-i)^k U_0}_{\HHN} \quad \and \quad \nr{\partial_t u(t)}_{L^2(\O)} \leq \frac c {t^{\frac d 4 + 1}} \nr{(\AcN-i)^k U_0}_{\HHN},
\]
where $u$ is the solution of \eqref{wave-neumann}.
\end{corollary}

Notice that for compactly supported initial data we do not have better estimates than for $(\AcN-i)^k U_0 \in \HHN^{-\d_2}$ with $\d_2 > \frac d 2$.

Finally, the dependance in $s_1$ in the estimates of Theorem \ref{th-energy-decay} emphasizes the fact the local energy decays faster than the global energy. We said that these two quantities should decay at the same speed for the contribution of high frequencies since the damping is effective at infinity. However this does not apply to the contribution of low frequencies, since then the damping term $a \partial_t u$ is small.

\begin{corollary}[Local energy decay]
Let $K$ be a compact subset of $\bar \O$. Let $k> d + 2$. Let $s$ be as in Theorem \ref{th-energy-decay} Then there exists $c \geq 0$ such that for $t \geq 1$ and $U_0 = (u_0,iu_1) \in \Dom(\AcN^k)$ supported in $K \times K$ we have 
\[
\nr{\nabla u(t)}_{L^2(K)} \leq \frac c {t^{\frac {d+1+s} 2}} \nr{(\AcN-i)^k U_0}_{\HHN} \quad \and \quad \nr{\partial_t u(t)}_{L^2(K)} \leq \frac c {t^{\frac d 2 + 1}} \nr{(\AcN-i)^k U_0}_{\HHN},
\]
where $u$ is the solution of \eqref{wave-neumann}.
\end{corollary}

If we can take $s = 1$ then we recover the same rate of local energy decay as in \cite{royer-diss-wave-guide}. This parameter $s$ will be discussed after Theorem \ref{th-chaleur} and in Remarks \ref{rem-tilde-s} and \ref{rem-s} below.\\

In all these statements, we had to deal simultaneously with the contributions of low and high frequencies, even if they have very different behaviors. The difficulty is that the operator $\AcN$ is not self-adjoint, so there is no obvious way to localize spectrally on low or high frequencies. One possibility is to localize with respect to the transverse (Neumann) Laplacian. This will be done in Theorem \ref{th-energy-decay-bis} under the additionnal assumption that the absoption index $a(x,y)$ only depends on $x$.\\

As already said, the contribution of low frequencies for the wave equation with damping at infinity is expected to behave like the solution of a corresponding heat equation. The purpose of the next result is to emphasize this fact. Before giving the statement, we remark that the low frequency part in the estimate of Theorem \ref{th-energy-decay} is exactly what we would obtain for the solution of a heat equation on $\R^d$. On the wave guide $\O$, we show that our solution behaves like a function which does not depend on $y \in \o$ and is indeed the solution of a heat equation with respect to $x \in \R^d$.\\

For $u \in L^2(\O)$ we set 
\begin{equation} \label{def-P0}
P_0 u : x \mapsto \frac 1 {\abs \o} \int_{y \in \o} u(x,y)\, dy.
\end{equation}
Then $P_0 u$ is defined for almost all $x \in \R^d$ and belongs to $L^2(\R^d)$. The function $P_0 u$ can also be seen as a function on $\O$ which does not depend on the transverse variable $y$, so that $P_0$ is a projection of $L^2(\O)$. We also set $P_0^\bot = 1- P_0$.\\

Now let $v$ be the solution on $\R_+ \times \R^d$ for the heat equation
\begin{equation} \label{heat}
\begin{cases}
\partial_t v - \D v = 0, & \text{on } \R_+ \times \R^d,\\
v(0) = P_0 (a u_0 + u_1) ,& \text{on } \R^d.
\end{cases}
\end{equation}
Again, this solution can be seen as a function on $\R_+ \times \O$ which does not depend on $y \in \o$.

\begin{theorem}[Comparison between the damped wave equation and the heat equation] \label{th-chaleur}
Let $s_1,s_2 \in \big[0,\frac d 2]$, $\k > 1$, $s \in [0,1]$, $\d_1 \geq  \k s_1 + s$ and $\d_2 \geq  \k s_2 + s$. Then there exists $C \geq 0$ such that for $t \geq 1$ and $u_0,u_1 \in L^{2,\d_2}(\O)$ we have 
\begin{equation} \label{estim-th-v}
\nr{\nabla v(t)}_{L^{2,-\d_1}(\O)} \leq C t^{-\frac 12 (1 + s_1 + s_2 + s)} \nr{au_0 + u_1}_{L^{2,\d_2}(\O)}
\end{equation}
and 
\begin{equation} \label{estim-th-dv}
\nr{\partial_t v(t)}_{L^{2,-\d_1}(\O)} \leq C t^{-\frac 12 (2 + s_1 + s_2)} \nr{au_0 + u_1}_{L^{2,\d_2}(\O)},
\end{equation}
where $v$ is the solution of \eqref{heat}. Now let $\tilde s \in [0, \min(2,d,\rho)[$. Then for $k \in \N^*$ there exists $C \geq 0$ such that for $t \geq 1$ and $U_0 = (u_0,iu_1) \in \Dom(\AcN^k)$ with $(\AcN-i)^k U_0 \in \HHN^{\d_2}$ we have 
\begin{equation} \label{estim-th-uv}
\nr{\nabla u(t) - \nabla v(t)}_{L^{2,-\k s_1}} \leq C t^{-\frac 12 (1 + s_1 + s_2 + \tilde s)} \nr{(\AcN-i)^k U_0}_{\HHN^{\k s_2}}
\end{equation}
and 
\begin{equation} \label{estim-th-duv}
\nr{\partial_t u(t) - \partial_t v(t)}_{L^{2,- \k s_1}} \leq C t^{-\frac 12 (2 + s_1 + s_2 + \tilde s)} \nr{(\AcN-i)^k U_0}_{\HHN^{\k s_2}},
\end{equation}
where $u$ is the solution of \eqref{wave-neumann}.
\end{theorem}

Theorem \ref{th-energy-decay} can be seen as a consequence of Theorem \ref{th-chaleur}. More precisely, if we can take $\tilde s$ greater than $s$ (this is for instance the case if we are interested in the global energy decay) then the energy of $u-v$ decays faster than that of $v$. This implies that $u$ behaves like $v$ at the first order for large times.

We notice that the restriction $s \leq 1$ in Theorem \ref{th-energy-decay} is due to the behavior of the solution of the heat equation (see Remark \ref{rem-s}), while the assumption $s < \min(d,\rho)$ comes from the analysis of the rest. We will see that in the case $a(x,y) \equiv 1$ we can take $s = 1$ even if $d \leq 2$. See Remark \ref{rem-tilde-s}.\\

Theorems \ref{th-energy-decay} and \ref{th-chaleur} are given for wave guides, which are the topic of this paper. However, for some aspects these results are better than what is known in the Euclidean space. First, we give optimal decay in any weighted space (for the initial conditions and for the energy itself), which is more precise than in the previous papers. Moreover, we allow a slow convergence of the absorption index to a constant. For these reasons it is important to notice that our analysis will also give these improvements in the Euclidean space. For $u_0 \in H^1(\R^d)$ and $u_1 \in L^2(\R^d)$ we consider the problem

\begin{equation} \label{wave-eucl}
\begin{cases}
\partial_t^2 u  -\D u + a \partial_t u = 0,  & \text{on  }  \R_+ \times \R^d, \\
\restr{(u , \partial_t u )}{t = 0} = (u_0, u_1), &  \text{on } \R^d,
\end{cases}
\end{equation}
where absorption index $a$ satisfies the same kind of estimate as on the wave guide:
\begin{equation} \label{hyp-amort-inf-Rd}
\forall x \in \R^d, \quad \abs{\partial_x^\b \big(a(x)-1\big)} \leq C_\b \pppg x^{-\rho-\abs \b}.
\end{equation}
for $\d \in \R$ we define $\HHRd^\d$ as $\HHN^\d$, except that the norms are in $L^2(\R^d)$ instead of $L^2(\O)$.\\

In this setting we obtain a result analogous to Theorem \ref{th-energy-decay}, except that we do not have any problem with high frequencies:

\begin{theorem} [Energy decay for the damped wave equation in the Euclidean space] \label{th-Rd}
Let $s_1,s_2 \in \big[0,\frac d 2]$, $\k > 1$, $s \in [0,\min(d,\rho)[$ with $s \leq 1$, $\d_1 \geq  \k s_1 + s$ and $\d_2 \geq  \k s_2 + s$. Then there exists $C \geq 0$ such that for $t \geq 1$ and $U_0 = (u_0,iu_1) \in \HHRd^{\d_2}$ we have 
\[
\nr{\nabla u(t)}_{L^{2,-\d_1}(\R^d)} \leq C t^{-\frac 12 (1 + s_1 + s_2 + s)} \nr{U_0}_{\HHRd^{\d_2}}
\]
and 
\[
\nr{\partial_t u(t)}_{L^{2,-\d_1}(\R^d)} \leq C t^{-\frac 12 (2 + s_1 + s_2 )} \nr{U_0}_{\HHRd^{\d_2}},
\]
where $u(t)$ is the solution of \eqref{wave-eucl}.
\end{theorem}

All these results will be proved from a spectral point of view. In Section \ref{sec-resolvent} we prove all the required resolvent estimates, in Section \ref{sec-energy-decay} we deduce the local and global energy decay for \eqref{wave-neumann} and, finally, in Section \ref{sec-dirichlet} we discuss some closely related problems: the above mentioned problem on $\R^d$, the problem similar to \eqref{wave-neumann} with Dirichlet boundary condition, and finally the damped Klein-Gordon equation in all these settings.\\

\section{Resolvent estimates} \label{sec-resolvent}

The proofs of Theorems \ref{th-energy-decay} and \ref{th-chaleur} rely on a spectral analysis (and in particular some resolvent estimates) for the operator $\AcN$ on $\EEN$ and for the corresponding Schr\"odinger operator on $L^2(\O)$.

\subsection{General properties}

Because of the damping, the operator $\AcN$ is not selfadjoint. However, since the absorption index $a$ has a sign, it is at least dissipative.

We recall that an operator $T$ on some Hilbert space $\Kc$ with domain $\Dom(T)$ is said to be dissipative if for all $\f \in \Dom(T)$ we have 
\[
\Im \innp{T\f} \f_\Kc \leq 0.
\]
Moreover $T$ is said to be maximal dissipative if it has no other dissipative extension on $\Kc$ than itself. We know that the dissipative operator $T$ is maximal dissipative if and only if $(T-\z)$ is boundedly invertible for some (therefore any) $\z \in \C_+$, where
\[
\C_+ : = \singl{\z \in \C \st \Im(\z) > 0}.
\]
Finally, the operator $T$ is said to be (maximal) accretive if $-iT$ is (maximal) dissipative. If $T$ is both dissipative and accretive, then it is maximal dissipative if and only if it is maximal accretive.\\

As usual for the damped wave equation, the resolvent $(\AcN-z)\inv$ on $\EEN$ will be expressed in terms of the resolvent of $-\DN- iza$, where $\DN$ is the Neumann realization of the Laplace operator on $\O$.

\begin{proposition}
Let $z \in \C_+$. Then $z^2$ belongs to the resolvent set of the operator $-\DN- iz a$.
\end{proposition}

\begin{proof}
The operator $-\DN$ is selfadjoint and non-negative on $L^2(\O)$. If $\Re(z) = 0$ then $-iza$ is a bounded and non-negative operator, so $-\DN-iza$ is selfadjoint and non-negative. Since $z^2$ is real negative, it belongs to its resolvent set. Now assume that $\Re(z) > 0$. Then $-iza$ is bounded and dissipative, so $-\DN-iza$ is maximal dissipative. Thus its resolvent set contains $\C_+$ and in particular $z^2$. Finally, if $\Re(z) < 0$ then $-(-\DN-iza)$ is maximal dissipative and $-z^2$ belongs to $\C_+$, so we can conclude similarly.
\end{proof}

For $z \in \C_+$ we set 
\[
\RNz = \big( - \DN  -iaz -z^2 \big)\inv.
\]

\begin{proposition} \label{prop-AcN-diss}
The operator $\AcN$ is maximal dissipative on $\EEN$. Moreover for $z \in \C_+$ and $F \in \HHN$ we have 
\begin{equation} \label{expr-res}
(\Ac_N-z)\inv F =
\begin{pmatrix}
\RNz (ia + z) & \RNz \\
I + \RNz (iza + z^2) & z \RNz
\end{pmatrix}F.
\end{equation}
\end{proposition}

\begin{proof}
For $U = (u,v) \in \Dom(\AcN)$ we have 
\[
\Im \innp{\AcN U}{U}_{\EEN} = - \innp{av}v_{L^2(\O)} \leq 0,
\]
so $\AcN$ is dissipative on $\EEN$. Then
\[
\nr{(\AcN-i)U}_{\EEN}^2 = \nr{\AcN}_{\EEN}^2 + \nr{U}_{\EEN}^2 - 2 \Im \innp{\AcN U}{U}_{\EEN} \geq \nr{\AcN}_{\EEN}^2 + \nr{U}_{\EEN}^2,
\]
so $(\AcN-i)$ is injective with closed range. It remains to prove that $\Ran(\AcN-i)$ is dense in $\EEN$. Let $F = (f,g) \in \HHN$. For $U =(u,v) \in \Dom(\AcN)$ we have 
\begin{align*}
(\AcN-i)U = F
& \eqv 
\begin{cases}
v-iu = f \\
- \DN u - i a v - i v= g
\end{cases}
\eqv
\begin{cases}
u = \RN(i) (g+iaf+if)\\
v = iu + f 
\end{cases}
\end{align*}
Defined this way, $U = (u,v)$ indeed belongs to $\Dom(\AcN)$ so $F \in \Ran(\AcN-i)$. Since $\HHN$ is dense in $\EEN$, this proves that $(\AcN-i)$ has a bounded inverse in $\Lc(\EEN)$, so $\AcN$ is maximal dissipative.
In particular any $z \in \C_+$ belongs to the resolvent set of $\AcN$. Then if we denote by $\Rc_\Ac(z) F$ the right-hand side of \eqref{expr-res}, we can check by straightforward computation that $\Rc_\Ac(z) F \in \Dom(\AcN)$ and
\[
(\AcN-z) \Rc_\Ac(z) F = F.
\]
This proves that $(\AcN - z)\inv = \Rc_\Ac(z)$ on $\HHN$.
\end{proof}

For the proofs of Theorems \ref{th-energy-decay} and \ref{th-chaleur} we have to estimate the resolvent $(\AcN-z)\inv$ when $\Im(z) > 0$ goes to 0. This aspect is simplified by the fact that with a strong absorption any real number except 0 belongs to the resolvent set of $\AcN$.

\begin{proposition} \label{prop-RNt}
Let $\t \in \R \setminus \singl 0$. Then the resolvent $\RNt$ is well defined and extends to an operator in $\Lc(\HuOp,\HuO)$.
\end{proposition}

\begin{proof}
Since the operator $-\DN  -i\t (a-1)$ is a bounded and relatively compact perturbation of the selfadjoint operator $-\DN$, we deduce by the usual Weyl Theorem (see \cite{rs4}, see also Appendix \ref{sec-weyl} for a discussion about the essential spectrum) that its essential spectrum is the same as for $-\DN$, namely $\R_+$. Then the essential spectrum of $-\DN -i\t a$ is $\R_+ -i\t$. In particular $\t^2 \in \R$ belongs to the spectrum of $-\DN -i\t a$ if and only if it is an eigenvalue. Now assume that $u \in \Dom(\DN)$ is such that $(-\DN -i\t a - \t^2)u = 0$. Then 
\begin{equation} \label{estim-im-part}
\int_\O a \abs u^2 = - \frac 1 \t \Im \innp{(-\DN -i\t a - \t^2)u }{u} = 0.
\end{equation}
This implies that $u$ vanishes where $a > 0$. Then $(-\DN - \t^2)u = 0$ and, by unique continuation, $u = 0$. Finally $\t^2$ belongs to the resolvent set of $-\D -i\t a$, which means that the resolvent $\RNt$ is well defined. It defines in particular a bounded operator from $L^2(\O)$ to $H^1(\O)$ and by duality, from $H^1(\O)'$ to $L^2(\O)$. By the resolvent identity 
\[
\RNt = (-\DN  + 1)\inv + \RNt (i\t a + \t^2 - 1) (-\DN  + 1)\inv,
\]
we conclude that $\RNt$ also extends to a bounded operator from $H^1(\O)'$ to $H^1(\O)$.
\end{proof}

\begin{proposition} \label{prop-inter-freq}
Any $\t \in \R \setminus \singl 0$ belongs to the resolvent set of $\AcN$.
\end{proposition}

\begin{proof}
Let $\t \in \R \setminus \singl 0$. Let $\m \in ]0,1]$ and $z = \t + i \m$. Let $U = (u,v) \in \HHN$. By Proposition \ref{prop-AcN-diss} we have 
\[
(\AcN-z)\inv U = 
\begin{pmatrix}
\frac 1 z \RNz (iza+z^2) u + \RNz v \\
u + \RNz (iza+z^2) u + z \RNz v
\end{pmatrix}
= 
\begin{pmatrix}
- \frac u z - \frac 1 z \RNz \DN u + \RNz v \\
-\RNz \DN u + z \RNz v
\end{pmatrix}.
\]
By Proposition \ref{prop-RNt} we have 
\[
\nr{ \nabla \RNz v}_{L^2(\O)} + \nr{z \RNz v}_{L^2(\O)} \lesssim \nr{v}_{L^2(\O)}
\]
(where $\lesssim$ means ``less or equal up to a multiplicative constant which does not depends on $\m \in ]0,1]$''). On the other hand, if we see $\DN$ as a bounded operator from $H^1(\O)$ to its dual we can write
\[
\nr{-z\inv \nabla u - z \inv \nabla \RNz \DN u}_{L^2(\O)} + \nr{\RNz \DN u}_{L^2(\O)} \lesssim \nr{\nabla u}.
\]
Thus 
\[
\nr{(\AcN-z)\inv U}_{\EEN} \lesssim \nr{U}_{\EEN}.
\]
This proves that $(\AcN -z)\inv$ is bounded in $\Lc(\EEN)$ uniformly in $\m \in ]0,1]$. Since this resolvent blows up near the spectrum of $\AcN$, this proves that $\t$ belongs to the resolvent set of $\AcN$.
\end{proof}

\begin{remark} \label{rem-expr-res}
For further use we remark that the computations for the resolvent $(\AcN-z)\inv$ still holds for $\t \in \R \setminus \singl 0$. Thus for $U = (u,v) \in \EEN$ we have 
\begin{align*}
\nr{(\AcN - \t)\inv U}_{\EEN}
& \leq \frac 1 {\abs{\t}} \nr{\nabla u} + \frac 1 {\abs{\t}} \nr{\nabla \RNt \nabla} \nr{\nabla u} + \nr{\nabla \RNt} \nr{v}\\
& + \nr{\RNt \nabla} \nr{\nabla u} + \abs \t \nr{\RNt} \nr {v}
\end{align*}
(all the norms on the right are in $L^2(\O)$ or $\Lc(L^2(\O))$).
\end{remark}

\begin{remark}
We can check that 0 belongs to the spectrum of $\AcN$.
\end{remark}

\subsection{Separation of variables} \label{sec-separation-variables}

On a wave guide $\O \simeq \R^d \times \o$, it is usual to write a Laplace operator as the sum of the Laplace operators on $\R^d$ and on $\o$. From the spectral properties of these two simpler terms we can deduce useful information about the full operator. See for instance \cite{KrejcirikKr05,borisovk08,KrejcirikRa14,royer-diss-wave-guide}.

Here we will have to be careful with the fact that the absorption index $a$ does not necessarily respect the symmetry of the domain $\O$. We will nonetheless use this idea in our proofs.\\

We denote by $-\Dx$ the usual Laplacian on $\R^d$. Then we denote by $\TN$ the Neumann realization of the Laplacian on $\o$. This is a non-negative selfadjoint operator on $L^2(\o)$ with compact resolvent, 0 being a simple eigenvalue thereof ($\TN\f = 0$ if and only if $\f$ is constant on $\o$). We denote by 
\[
0 = \l_0 < \l_1 \leq \dots \leq \l_k \leq \dots
\]
the eigenvalues of $\TN$ and consider a corresponding orthonormal basis $\seq \f k$ of eigenfunctions: for $k \in \N$ we have $\nr{\f_k}_{L^2(\o)} = 1$, $\f_k \in \Dom(\TN)$ and $\TN \f_k = \l_k \f_k$.

Let $u \in L^2(\O)$. For almost all $x \in \R^d$ we have $u(x,\cdot) \in L^2(\o)$, so there exists a sequence $(u_k(x))_{k \in \N}$ such that in $L^2(\o)$ we have
\begin{equation} \label{def-uk}
u(x,\cdot) = \sum_{k\in\N} u_k(x) \f_k
\end{equation}
and
\[
\nr{u(x,\cdot)}_{L^2(\o)}^2 = \sum_{k \in \N} \abs{u_k(x)}^2.
\]
After integration over $x \in \R^d$ we obtain that $u_k \in L^2(\R^d)$ for all $k \in \N$ and 
\[
\nr{u}_{L^2(\O)}^2 = \sum_{k\in\N} \nr{u_k}_{L^2(\R^d)}^2.
\]

Let $\a$ be a bounded function on $\R^d$. We can see $\a$ as a fonction on $\O$ which does not depend on $y \in \o$. It is not difficult to check that for $\th \in \s(-\Dx -i\a)$ and $k \in \N$ we have $\th + \l_k \in \s(-\DN -i\a)$ (see for instance Proposition 4.1 in \cite{royer-diss-schrodinger-guide} in a similar context). Here we have denoted by $\Dx$ the usual Laplacian on $\R^d$. Let $\z \in \C$ in the resolvent set of $-\DN -i\a$. Let $f \in L^2(\o)$. As above we can write  
\begin{equation} \label{def-fk}
f = \sum_{k\in\N} f_k \otimes \f_k,
\end{equation}
where $\seq f k$ is a sequence of functions in $L^2(\R^d)$. 
For $m \in \N$ we set 
\[
g_m = \sum_{k = 0}^m f_k \otimes \f_k \quad \text{and} \quad v_m = \sum_{k =0}^m \big(-\Dx  - i \a(x) - (\z - \l_k) \big)\inv f_k \otimes \f_k.
\]
Then for all $m \in \N$ we have $v_m \in \Dom(\DN)$ and by direct computation $(-\DN -i\a(x) - \z) v_m = g_m$. Therefore $v_m = (-\DN -i\a(x) - \z) \inv g_m$. At the limit $m \to \infty$ we obtain 
\begin{equation} \label{res-sep-variables}
(-\DN-i\a(x)-\z)\inv f = \sum_{k \in \N} \big(-\Dx  - i \a(x) - (\z - \l_k) \big)\inv f_k \otimes \f_k.
\end{equation}

In Section \ref{sec-low-freq} we will show that the resolvent $\RNz$ is close to $(-\Dx-iz)\inv$ when $z \in \C_+$ is small. For this we will use some estimates of the latter. Since the absorption index has been replaced by 1, we can use the separation of variables.

For $\s \in \R$ we denote by $\dot H^\s L^2(\O)$ the homogeneous Sobolev space with respect to $x \in \R^d$, endowed with the natural norm 
\begin{equation} \label{def-Hs}
\nr{u}_{\dot H^\s L^2(\O)}^2 = \int_{y \in \o} \nr{u(\cdot,y)}_{\dot H^\s(\R^d)}^2 \, dy = \nr{(-\Dx)^{\frac \s 2} u}_{L^2(\O)}^2.
\end{equation}

\begin{proposition} \label{prop-res-heat}
Let $\s \in \R$ and $j \in \N^*$. For $z \in \C_+$ we have 
\[
\nr{(-\DN -iz)^{-j}}_{\Lc(\dot H^\s L^2(\O))} = \frac {1}{\abs z^j}, \qquad \nr{\nabla_x (-\DN -iz)^{-j}}_{\Lc(\dot H^\s L^2(\O))} \leq \frac {\sqrt 2}{\abs z^{j-\frac 12}},
\]
and
\[
\nr{\nabla_y (-\DN -iz)^{-j}}_{\Lc(\dot H^\s L^2(\O))} \leq \l_1^{\frac 12 -j}.
\]
\end{proposition}

\begin{proof}
Since $(-\DN -iz)\inv$, $\nabla_x$ and $\nabla_y$ commute with $(-\D_x)^{\frac \s 2}$, it is enough to consider the case $\s = 0$. The first estimate comes from the facts that $-\DN$ is selfadjoint and non-negative and that for $z \in \C_+$ we have $d(iz,\R_+) = \abs z$. For $f\in L^2(\O)$ and $\seq f k$ as in \eqref{def-fk} we have by \eqref{res-sep-variables}
\[
-\Dx (-\DN -iz)\inv f = \sum_{k = 0}^\infty -\Dx (-\Dx - iz + \l_k)\inv f_k \otimes \f_k.
\]
Since $iz - \l_k$ has negative real part
\begin{align*}
\nr{-\Dx (-\DN -iz)\inv f}_{L^2(\O)}^2
& \leq \sum_{k = 0}^\infty \nr{f_k + (iz-\l_k) (-\Dx - iz + \l_k)\inv f_k}_{L^2(\R^d)}^2\\
& \leq 4 \sum_{k=0}^\infty \nr{f_k}_{L^2(\R^d)}^2 = 4 \nr{f}_{L^2(\O)}^2.
\end{align*}
Then 
\begin{align*}
\nr{\nabla_x (-\DN -iz)\inv f}_{L^2(\O)}^2 = {\innp{-\Dx (-\DN -iz)\inv f}{(-\DN -iz)\inv f}} \leq \frac 2 {\abs z} \nr f_{L^2(\O)}^2,
\end{align*}
and the second estimate follows.
For the last estimate we recall that $\f_0$ is constant, so 
\[
\nabla_y (-\DN -iz)^{-j} f = \sum_{k = 1}^\infty  (-\Dx - iz + \l_k)^{-j} f_k \otimes \nabla_y \f_k.
\]
For $k,l \in \N$ we have
\[
\innp{\nabla_y \f_k}{\nabla_y \f_l}_{L^2(\o)} = \innp{-\D_y \f_k}{\f_l}_{L^2(\o)} = \l_k \d_{k,l},
\]
so the family $(\nabla_y \f_k)$ is orthogonal and $\nr{\nabla_y \f_k}_{L^2(\O)}^2 = \l_k$ for all $k \in \N$. This gives 
\[
\nr{ \nabla_y (-\DN -iz)^{-j} f}_{L^2(\O)}^2 = \sum_{k=1}^\infty \l_k \nr{(-\Dx - iz + \l_k)^{-j} f_k}_{L^2(\R^d)}^2 \leq \l_1^{1-2j} \nr{f}_{L^2(\O)}^2
\]
and concludes the proof.
\end{proof}

\subsection{Contribution of high frequencies} \label{sec-high-freq}

In this section we study the resolvent $(\AcN-\t)\inv$ for $\abs \t \gg 1$. We already know from Proposition \ref{prop-inter-freq} that any $\t \in \R \setminus \singl 0$ belongs to the resolvent set of $\AcN$. In the following theorem we give an estimate for the resolvent on the real axis at infinity. This gives in particular a region around the real axis free of spectrum.

\begin{theorem} \label{th-high-freq}
There exist $\t_0 > 0$, $\g > 0$ and $C \geq 0$ such that any $z \in \C$ with $\abs{\Re(z)} \geq \t_0$ and $\Im(z) \geq - \g \abs{\Re(z)}^{-2}$ belongs to the resolvent set of $\AcN$ and 
\[
\nr{(\Ac-z)\inv}_{\Lc(\EEN)} \leq C \abs {\Re(z)}^2.
\]
\end{theorem}

By Remark \ref{rem-expr-res}, Theorem \ref{th-high-freq} is a direct consequence of the following result (with $H^0(\O) \simeq H^0(\O)' := L^2(\O)$):

\begin{proposition} \label{prop-RNt-high-freq}
Let $\b_1,\b_2 \in \{0,1\}$. There exist $\t_0 > 0$ and $c \geq 0$ such that for any $\t \in \R \setminus [-\t_0,\t_0]$ we have 
\[
\nr{\RNt}_{\Lc(H^{\b_2}(\O)',H^{\b_1}(\O))} \leq c\t^{1+\b_1+\b_2}.
\]
\end{proposition}

For the proof of Proposition \ref{prop-RNt-high-freq} we use the same kind of ideas as in \cite{Burq-Hi-07}. After a separation of variables, the problem will reduce to a similar problem on the Euclidean space. We will need the following estimate, which will be discussed in Appendix \ref{sec-high-freq-Rd}:

\begin{proposition} \label{prop-estim-res-Rd}
Let $\a$ be a non-negative valued function on $\R^d$ such that $\a \geq c_0$ for some $c_0 > 0$ outside some bounded subset of $\R^d$. Then there exist $\t_0 \geq 0$ and $c \geq 0$ such that for $\abs \t \geq \t_0$ we have 
\[
\nr{(-\Dx  - i\t \a - \t^2)\inv}_{\Lc(L^2(\R^d))} \leq \frac c \t.
\]
\end{proposition}

With this estimate we can prove Proposition \ref{prop-RNt-high-freq}:

\begin{proof} [Proof of Proposition \ref{prop-RNt-high-freq}]

\stepp Assume that we have proved that for $\t \in \R \setminus \singl 0$, $f \in L^2(\O)$ and $u = \RNt f$ we have 
\begin{equation} \label{estim-BH}
\nr{u}_{L^2(\O)}^2 \lesssim \nr{f}_{L^2(\O)}^2 + \t^2 \int_\O a\abs u^2.
\end{equation}
By \eqref{estim-im-part} we have 
\[
\t^2 \int_\O a \abs u^2 \leq \t \nr f_{L^2(\O)} \nr u_{L^2(\O)},
\]
so there exists $C \geq 0$ which does not depend on $\t$, $f$ or $u$ such that
\[
\nr{u}_{L^2(\O)}^2 \leq \frac 12 \nr{u}_{L^2(\O)}^2 + C \t^2 \nr{f}_{L^2(\O)}^2.
\]
This yields the result when $(\b_1,\b_2) = (0,0)$. 
For $\vf \in L^2(\O)$ we have 
\begin{equation} \label{nabla-RNt}
\nr{\nabla \RNt \vf}^2 = \innp{\vf}{\RNt \vf} + i \t \innp{a \RNt \vf}{\RNt \vf} + \t^2 \nr{\RNt \vf}^2.
\end{equation}
This gives the case $(\b_1,\b_2) = (1,0)$. The case $(\b_1,\b_2) = (0,1)$ follows by duality, and for the case $(\b_1,\b_2) = (1,1)$ we use \eqref{nabla-RNt} again. Thus the conclusion will follow from the proof of \eqref{estim-BH}.

\stepp First assume that \eqref{estim-BH} is proved when the absorption index only depends on $x \in \R^d$. There exists $\a$ such that $0 \leq \a \leq a$ on $\O$, $\a \geq c_0 >0$ outside some bounded subset of $\O$ and $\a$ only depends on $x \in \R^d$. We have 
\[
(-\D -i\t \a - \t^2) u = f + i\t(a-\a) u.
\]
By \eqref{estim-BH} applied with $a$ replaced by $\a$ we get 
\[
\nr{u}_{L^2(\O)}^2 \lesssim \nr{f}_{L^2(\O)}^2 + \t^2 \nr{(a-\a) u}_{L^2(\O)}^2 + \t^2 \int_\O \a \abs u^2.
\]
Since $a - \a$ is bounded, this yields
\[
\nr{u}_{L^2(\O)}^2 \lesssim \nr{f}_{L^2(\O)}^2 + \t^2 \int_\O (a-\a) \abs u^2 + \t^2 \int_\O \a \abs u^2.
\]
This gives \eqref{estim-BH} for $a$. 

\stepp It remains to prove \eqref{estim-BH} when $a$ only depends on $x$. In this case we can apply the separation of variables of Section \ref{sec-separation-variables}. 
Let $\seq u k$ and $\seq f k$ be the sequences in $L^2(\R^d)$ defined as in \eqref{def-uk} and \eqref{def-fk}. 
As in the proof of Proposition \ref{prop-RNt} we observe that the real axis is included in the resolvent set of the operator $-\Dx - i\t a$. In particular \eqref{res-sep-variables} applied with $\a = \t a$ holds for any $\z \in \R$, so for all $k \in \N$ we have in $L^2(\R^d)$
\[
u_k = \big( -\Dx  -i\t a - (\t^2-\l_k) \big)\inv f_k.
\]
Let $k\in\N$. The operator $-\Dx -i\t a$ is maximal accretive (if $\t > 0$ it is accretive and maximal dissipative, and if $\t < 0$ we consider its adjoint) so if $\t^2 - \l_k \leq - 1$ we have
\[
\nr{u_k}_{L^2(\R^d)} \leq \nr{f_k}_{L^2(\R^d)}.
\]
Let $\t_0$ be given by Proposition \ref{prop-estim-res-Rd}. By continuity of the resolvent $\z \mapsto (-\Dx -i \t a - \z)\inv$ there exists $C \geq 0$ such that if $\t^2 - \l_k \in [-1,\t_0^2]$ we have 
\[
\nr{u_k}_{L^2(\R^d)} \leq C \nr{f_k}_{L^2(\R^d)}.
\]
It remains to consider the case $\t^2 - \l_k \geq \t_0^2$. Let $\s = \sqrt {\t^2 - \l_k} \in [\t_0,\t]$. We have 
\[
(-\Dx - i\s a - \s^2) u_k = f_k + i (\t-\s)a u_k.
\]
By Proposition \ref{prop-estim-res-Rd} we obtain 
\[
\nr{u_k}_{L^2(\R^d)} \lesssim \nr{f_k}_{L^2(\R^d)} + \abs \t \nr{au_k}_{L^2(\R^d)} \lesssim \nr{f_k}_{L^2(\R^d)} + \abs \t \nr{\sqrt a u_k}_{L^2(\R^d)}.
\]
Finally 
\begin{align*}
\nr{u}_{L^2(\O)}
& = \sum_{k\in\N} \nr{u_k}_{L^2(\R^d)}^2 \lesssim \sum_{k\in\N} \nr{f_k}_{L^2(\R^d)}^2 + \t^2 \nr{\sqrt a u_k}_{L^2(\R^d)}^2 \lesssim \nr{f}_{L^2(\O)}^2 + \t^2 \nr{\sqrt a u}_{L^2(\O)}^2.
\end{align*}
This is \eqref{estim-BH}, and the proof of the proposition is complete. 
\end{proof}

\begin{remark} \label{rem-triv-high-freq}
If $a \equiv 1$ then in particular GCC holds and we easily get better high frequency resolvent estimates. If $\t \in \R \setminus \singl 0$, $u \in H^2(\O)$ and $f \in L^2(\O)$ are such that 
\[
(-\D -i\t - \t^2) u = f,
\]
then
\[
\t \nr{u}_{L^2(\O)}^2 = - \Im \innp{f}{u} \leq \nr f \nr u,
\]
from which we deduce that 
\[
\nr{\RNt}_{\Lc(L^2(\O)} \lesssim \frac 1 \t
\]
and consequently that $(\AcN-\t)\inv$ is uniformly bounded for $\abs \t \geq 1$.
\end{remark}

\subsection{Contribution of low frequencies.} \label{sec-low-freq}

In this section we study the resolvent $(\AcN-z)\inv$ when $z$ is close to 0. As said in introduction, for large times the solution $u(t)$ of \eqref{wave-neumann} behaves like the solution of \eqref{heat} if the initial condition is regular enough. And this is due to the contribution of low frequencies.
This can be seen from the resolvent point of view. More precisely, we prove in Theorem \ref{th-low-freq} than for $z$ small and in a suitable sense the resolvent $(\AcN-z)\inv$ is close to
\begin{equation} \label{def-RcC}
\RcC := 
\begin{pmatrix}
i \Rcz P_0 a & \Rcz P_0 \\
iz \Rcz P_0 a & z\Rcz P_0
\end{pmatrix},
\end{equation}
where $\Rcz$ is the resolvent corresponding to the heat equation (the projection $P_0$ appears in the initial data in \eqref{heat}, and the second row in $\RcC$ corresponds to the time derivative, which explains the extra factor $z$).\\

In order to prove the estimates on $u-v$ in Theorem \ref{th-chaleur}, we need estimates for the difference between $(\AcN-z)\inv$ and $\RcC$ for $z$ small. When it has a sense we set 
\begin{equation} \label{def-Th}
\Th(z)  = 
\begin{pmatrix}
\Th_1(z) & \Th_2(z) \\
\Th_3(z) & \Th_4(z)
\end{pmatrix}
:= (\AcN-z)\inv - \RcC.
\end{equation}

\begin{theorem} \label{th-low-freq}
Let $s_1,s_2 \in \big[0,\frac d 2\big]$, $\d_1 > s_1$ and $\d_2 > s_2$. Let $s \in ]0,\min(2,\rho,d)[$. Let $m \in \N$. Then there exists $C \geq 0$ such that for $j \in \{1,2\}$ and $z \in \C_+$ with $\abs z \leq 1$ we have 
\begin{equation} \label{estim-Th12}
\nr{\pppg x^{-\d_1} \nabla \Th_j^{(m)}(z) \pppg x^{-\d_2}}_{\Lc(L^{2}(\O))} \leq C \left( 1 + \abs z^{\frac 12 (s_1 + s_2 + s -1) - m }\right),
\end{equation} 
and for $j \in \{3.4\}$:
\begin{equation} \label{estim-Th34}
\nr{\pppg x^{-\d_1} \Th_j^{(m)}(z) \pppg x^{-\d_2}}_{\Lc(L^{2}(\O))} \leq C \left( 1 + \abs z^{\frac 12 (s_1 + s_2 + s) - m }\right).
\end{equation}
\end{theorem}

The proof of Theorem \ref{th-low-freq} is based on a scaling argument as in \cite{bouclet11,boucletr14,royer-dld-energy-space}.
For $\s \in \R$ we have defined the partial homogeneous Sobolev space $\dot H^\s L^2(\O)$ in \eqref{def-Hs}. We similarly define $H^\s L^2(\O)$, endowed with the norm 
\[
\nr{u}_{H^\s L^2(\O)}^2 = \int_{y \in \o} \nr{u(\cdot,y)}_{H^\s(\R^d)}^2 \, dy = \nr{(-\Dx + 1)^{\frac \s 2} u}_{L^2(\O)}^2.
\]
For $z \in \C_+$, $u \in C_0^\infty(\bar \O)$ and $(x,y) \in \O = \R^d \times \o$ we set 
\[
(\Phi_z u)(x,y) = \abs{z}^{\frac d 4} u \big( \abs z^{\frac 12} x, y \big)
\]
(notice that $\Phi_z$ only depends on $\abs z$). The function $\Phi_z$ extends to a bounded map on $\dot H^\s L^2(\O)$ for any $\s \geq 0$ and we have
\begin{equation} \label{nr-Phiz-1}
\nr{\Phi_z}_{\Lc(\dot H^\s L^2(\O))} = \abs{z}^{\frac d 4 + \frac \s 2}
\end{equation}
(we can see this directly for $\s \in \N$ and by interpolation for the general case). We also have 
\begin{equation} \label{nr-Phiz-2}
\nr{\Phi_z^*}_{\Lc(\dot H^{-\s} L^2(\O))} = \nr{\Phi_z\inv}_{\Lc(\dot H^{-\s} L^2(\O))} = \nr{\Phi_{z\inv}}_{\Lc(\dot H^{-\s} L^2(\O))} = \abs{z}^{-\frac d 4 + \frac \s 2}.
\end{equation}

The following proposition generalizes the idea that, as in the Hardy inequality, the multiplication by a decaying function behaves in some sense like a derivative, which is small for low frequencies.

\begin{proposition}[Proposition 7.2 in \cite{boucletr14}] \label{prop-dec-sob}
Let $s \geq 0$, $\rho > s$ and $\s \in \big] -\frac d 2, \frac d 2\big[$ be such that $\s -s \in \big] -\frac d 2, \frac d 2\big[$. Let $m \in \N$ be greater than $\frac d 2$. Let $\vf \in C^m(\R^d)$ be such that 
\[
\nr{\vf}_\rho := \sup_{\abs \b \leq m} \sup_{x \in \R^d} \abs{\pppg x^{\rho + \abs \b} \partial^\b \vf(x)} < + \infty.
\]
Then there exists $C \geq 0$ (which does not depend on $\vf$) such that for $\t > 0$ we have
\[
 \nr{\Phi_\t \inv \vf \Phi_\t}_{\Lc(H^\s(\R^d),H^{\s-s}(\R^d))} \leq C \t^{\frac s 2} \nr \vf_{\rho}.
\]
\end{proposition}

Notice that we can similarly obtain estimates in $\Lc(H^\s L^2(\O),H^{\s-s}L^2(\O))$ if we apply this result for each fixed $y$ and integrate the estimate over $\o$.\\

For $z \in \C_+$ we set $\hat z = z / \abs z$, $\hat a_{z} =\Phi_z\inv a \Phi_z$, $a_0 = a-1$, $\hat a_{0,z} = \Phi_z\inv (a-1) \Phi_z$,
\[
-\Dz = - \frac 1 {\abs z} \Phi_z\inv \DN  \Phi_z =  -\Dx + \frac {T_N}{\abs z},
\]
and 
\[
\hRNz = \big( -\Dz - i \hat z \hat a_{z} - z \hat z \big)\inv.
\]
Then we have 
\[
\RNz = \frac 1 {\abs z} \Phi_z \hRNz \Phi_z \inv.
\]

\begin{proof}[Proof of Theorem \ref{th-low-freq}]
\stepp We first remark that it is enough to prove the result when $s_1,s_2 \in \big[0,\frac d 2 \big[$. If this is not the case, then we can apply the result with $s_1$, $s_2$ and $s$ replaced by $s_1 - \e_1$, $s_2 - \e_2$ and $s + \e_1 + \e_2$ where $\e_1,\e_2 \geq 0$ are such that $s_1 -\e_1$ and $s_2 - \e_2$ belong to $\big[0,\frac d 2 \big[$, and $s + \e_1 + \e_2 < \min(2,d,\rho)$.

\stepp We begin with the estimate of
\begin{equation} \label{def-Th1}
\Th_1(z) = i (-\DN - iz)\inv P_0^\bot a + i \big( \RNz - \Rcz \big) a + z \RNz.
\end{equation}
By \eqref{res-sep-variables} applied with $\a = 0$ we see that $\nabla (-\DN - iz)^{-1-m} P_0^\bot$ extends to a holomorphic function on a neighborhood of 0. In particular all its derivatives are uniformly bounded on a neighborhood of 0 and hence satisfy the estimates of the proposition.

\stepp For $z \in \C_+$ the resolvent identity gives 
\[
\RNz - \Rcz = iz \Rcz a_0 \RNz + z^2 \Rcz \RNz.
\]
In \eqref{def-Th1} we replace $z \RNz$ by 
\[
z \Rcz + iz^2 \Rcz a_0 \RNz + z^3 \Rcz \RNz. 
\]
This ensures that in all the terms which remain to estimate in \eqref{def-Th1} the first factor is $\Rcz$. Then, considering their derivatives of order $m$ and taking into account the gradient in \eqref{estim-Th12}, we see that we have to estimate a linear combination of terms of the form 
\begin{equation} \label{def-T}
T(z) = z^q \pppg x^{-\d_1} \partial^\b (-\DN-iz)^{-j} a_0^{\nu}  \prod_{l=1}^k \big((ia+2z)^{\th_l} \RNz \big) \pppg x^{-\d_2},
\end{equation}
where $q \in \N$, $\b = (\b_x,\b_y) \in \N^{d+n}$ (with $\abs \b = 1$, $\b_x \in \N^d$, $\b_y \in \N^n$), $j \in \N^*$, $\n \in \{0,1\}$, $k \geq \n$, $\th_1 = 0$, $\th_2,\dots,\th_k \in \{0,1\}$ and
\begin{equation} \label{ineq-m}
j + k - q - \n \leq m.
\end{equation}
For this we proceed by induction on $m$ (this is true when $m = 0$ and the left-hand side cannot increase by more than 1 for each derivation). Notice that we can forget the factor $a$ which appears in the second term of the right-hand side of \eqref{def-Th1} since it is a bounded operator which commutes with the weight $\pppg x^{-\d_2}$.

\stepp 
Let $\Zc$ be a compact of $\C$ which does not intersect $\R_+$ and is a neighborhood of $\singl{-e^{i\th}, \abs{\th} \leq \frac \pi 2}$.
Let $\s \in \big] - \frac d 2, \frac d 2 \big[$. 
As in the proof of Proposition \ref{prop-res-heat} we can check that the resolvent $(-\Dz  - \z )\inv$ is bounded in $\Lc(H^{\s-1}L^2(\O),H^{\s+1}L^2(\O))$ uniformly in $\z \in \Zc$ and $z \in \C_+$.
Let
\[
\s_1 , \s_2 \in \left] \max\left( - \frac d 2 , \s-1 \right) , \min \left( \frac d 2 , \s+1 \right) \right[
\]
be such that $\s_1 - \s_2 = s$.
By \eqref{hyp-amort-inf} and Proposition \ref{prop-dec-sob} there exists $C \geq 0$ such that for $z \in \C_+$ with $\abs z \leq 1$ and $\z \in \Zc$ we have 
\begin{multline*}
\nr{ (-\Dz  - \z) - (-\Dz  - i \hat z \hat a_{0,z} - \z)} _{\Lc(H^{\s+1}L^2(\O),H^{\s-1}L^2(\O))}\\
= \nr{\hat a_{0,z}} _{\Lc(H^{\s+1}L^2(\O),H^{\s-1}L^2(\O))}
\leq \nr{\hat a_{0,z}} _{\Lc(H^{\s_1}L^2(\O),H^{\s_2}L^2(\O))} \leq C \abs  z^{\frac s 2}.
\end{multline*}
Since $(-\Dz-\z)$ is an isomorphism from $H^{\s+1}L^2(\O)$ to $H^{\s-1}L^2(\O)$ with inverse uniformly bounded in $\z$, we deduce that for $z \in \C_+$ small enough
\[
\nr{(-\Dz  - i \hat z \hat a_{0,z} - \z)\inv}_{\Lc(H^{\s-1}L^2(\O),H^{\s+1}L^2(\O))} \lesssim 1.
\]
With $\z = i \hat z + z \hat z$ (which belongs to $\Zc$ for $z \in \C_+$ small enough) we obtain
\begin{equation} \label{estim-hatRN}
\nr{\hRNz}_{\Lc(H^{\s-1}L^2(\O),H^{\s+1}L^2(\O))} \lesssim 1.
\end{equation}

\stepp 
Now let $T(z)$ be a term of the form \eqref{def-T}. 
Let $\s_1 \in [0,s_1]$ and $\s_2 \in [0,s_2]$ be such that
\[
\s_1 + \s_2 = \min\big( 2j + 2k - \n s - \abs{\b_x}, s_1 + s_2 \big).
\]
Since $\d_1 > \s_1$, we obtain by the Sobolev embeddings that the weight $\pppg x^{-\d_1}$ defines a continuous operator from $\dot H^{\s_1} L^2(\O)$ to $L^2(\O)$. Similarly, $\pppg x^{-\d_2} \in \Lc(L^2(\O), \dot H^{-\s_2}L^2(\O))$. Moreover, since $\s_1,\s_2 \geq 0$ we have $H^{\s_1}L^2(\O) \subset \dot H^{\s_1} L^2(\O)$ and $\dot H^{-\s_2}L^2(\O) \subset H^{-\s_2} L^2(\O)$.
Thus 
\[
\nr{T(z)} \lesssim \abs z^{q} \nr{\partial^{\b} (-\DN  -i z)^{-j}}_{\Lc(\dot H^{\s_1}L^2(\O))} \nr{ a_0^\nu \prod_{l=1}^k \big( (ia + 2z)^{\th_l} \RNz \big)}_{\Lc(\dot H^{-\s_2}L^2(\O) , \dot H^{\s_1}L^2(\O))}.
\]
By Proposition \ref{prop-res-heat} we obtain
\[
\nr{T(z)} \lesssim \abs z^{q + \frac 12 - j - k }\nr{\Phi_z \hat a_{0,z}^\nu \prod_{l=1}^k \big( (i\hat a_z + 2z)^{\th_l} \hRNz \big) \Phi_z\inv}_{\Lc(\dot H^{-\s_2}L^2(\O) , \dot H^{\s_1}L^2(\O))}.
\]
By \eqref{nr-Phiz-1}-\eqref{nr-Phiz-2}, \eqref{estim-hatRN}, Proposition \ref{prop-dec-sob} for $a_{0,z}$ and finally \eqref{ineq-m} this yields
\[
\nr{T(z)} \lesssim \abs z^{q + \frac 12 - j - k + \frac {\s_1 + \s_2 + \n s}2} \lesssim 1 + \abs z ^{\frac {s_1+s_2 +s \n} 2 + \frac 12 - \n - m}.
\]
This gives \eqref{estim-Th12}. We proceed similarly for \eqref{estim-Th34}, except that there is no derivative and \eqref{ineq-m} is replaced by 
\[
1 + j + k - q - \n \leq m.
\]
The proof is complete.
\end{proof}

\begin{remark} \label{rem-tilde-s}
We see in the proof that the restriction $\tilde s < \min(d,\rho)$ in Theorem \ref{th-chaleur} comes from the restrictions in Proposition \ref{prop-dec-sob} applied to $a_0$ (in particular we can remove this assumption if $a_0 = 0$). The restriction $\tilde s \leq 2$ is due to the terms for which there is no factor $a_0$ but an extra power of $z$.
\end{remark}

\section{Energy decay} \label{sec-energy-decay}

In this section we use the resolvent estimates of Proposition \ref{prop-inter-freq} and Theorems \ref{th-high-freq} and \ref{th-low-freq} to prove Theorems \ref{th-energy-decay} and \ref{th-chaleur}. 

Because of the singularity at 0 for the resolvent and since we work in weighted spaces, we cannot use the abstract results of \cite{BorichevTo10} or \cite{Batty-Ch-To} to convert resolvent estimates for $\AcN$ into estimates for the propagator (this will be possible in the next section for the similar problem with Dirichlet boundary conditions since then 0 is not in the spectrum). Here, as in \cite{Burq-Hi-07}, the strategy is inspired by \cite{lebeau96}. In these two papers there is no singularity at 0 and no weighted spaces, so we have to adapt the idea to take into account low frequencies.

\subsection{The Heat equation}

We begin with the decay estimates for the solution $v$ of the heat equation \eqref{heat}. Estimates \eqref{estim-th-v} and \eqref{estim-th-dv} are consequences of the following proposition:

\begin{proposition} \label{prop-heat-decay}
Let $s_1,s_2 \in \big[ 0, \frac d 2 \big]$ and $\k > 1$. Let $\b \in \N^d$ with $\abs \b \leq 1$ and $s \in [0,\abs \b]$. Then there exists $C \geq 0$ such that for all $t \geq 1$ we have 
\[
\nr{\pppg x^{-\k s_1 - s} \partial^\b e^{t\Dx } \pppg x^{- \k s_2 - s}}_{\Lc(L^2(\R^d))} \leq C \, t^{-\frac {s_1 + s_2 + \abs \b + s} 2}
\]
and
\[
\nr{\pppg x^{-\k s_1} \Dx e^{t\Dx} \pppg x^{-\k s_2}}_{\Lc(L^2(\R^d))} \leq C \, t^{-\frac {s_1 + s_2 + 2} 2}.
\]
\end{proposition}

\begin{proof}
We recall that the kernel of the heat propagator $e^{t\Dx }$ is given by
\[
K(t,x) = \frac 1 {(4\pi t)^{\frac d 2}} e^{-\frac {\abs x^2}{4t}}.
\]
For $t > 0$, $\vf \in C_0^\infty(\R^d)$ and $x \in \R^d$ we have 
\begin{eqnarray*}
\lefteqn{\nr{\pppg x^{-s} \partial^\b e^{t\Dx} \pppg x^{-s} \vf}_{L^2(\R^d)}^2}\\
&& = \int_{x\in\R^d} \pppg x^{-2s} \left( \int_{\tilde x \in \R^d} \frac 1 {(4\pi t)^{\frac d 2}} \frac {(x-\tilde x)^\b}{(2t)^{\abs \b}} e^{-\frac {\abs {x - \tilde x}^2}{4t}} \pppg {\tilde x}^{-s} \vf(\tilde x) \, d\tilde x\right)^2 \, dx\\
&& \leq \int_{x\in\R^d} \left( \int_{\tilde x \in \R^d} \frac 1 {(4\pi t)^{\frac d 2}} \frac {\pppg x^{-s} \abs{x-\tilde x}^{s}\pppg {\tilde x}^{-s}}{(2t)^{\frac {\abs \b + s}2}} \frac {\abs{x-\tilde x}^{\abs \b-s}}{(2t)^{\frac {\abs \b - s}2}} e^{-\frac {\abs {x - \tilde x}^2}{4t}}  \abs{\vf(\tilde x)} \, d\tilde x\right)^2 \, dx\\
&& \lesssim  \frac 1 {t^{d + \abs \b + s}} \int_{x\in\R^d} \left( \int_{\tilde x \in \R^d}  \frac {\abs{x-\tilde x}^{\abs \b-s}}{(2t)^{\frac {\abs \b - s}2}} e^{-\frac {\abs {x - \tilde x}^2}{4t}}  \abs{\vf(\tilde x)} \, d\tilde x\right)^2 \, dx\\
&& \lesssim  \frac {\nr{\vf}_{L^1(\R^d)}} {t^{d + \abs \b + s}} \int_{x\in\R^d} \int_{\tilde x \in \R^d}  \frac {\abs{x-\tilde x}^{2(\abs \b-s)}}{(2t)^{{\abs \b - s}}} e^{-\frac {\abs {x - \tilde x}^2}{2t}}  \abs{\vf(\tilde x)} \, d\tilde x \, dx\\
&& \lesssim  \frac {\nr{\vf}_{L^1(\R^d)}} {t^{d + \abs \b + s}} \int_{\tilde x \in \R^d}  \abs{\vf(\tilde x)} \int_{x\in\R^d} \frac {\abs{x-\tilde x}^{2(\abs \b-s)}}{(2t)^{{\abs \b - s}}} e^{-\frac {\abs {x - \tilde x}^2}{2t}}  \, dx  \, d\tilde x \\
&& \lesssim \frac {\nr{\vf}_{L^1(\R^d)}^2} {t^{\frac d 2 + \abs \b + s}}.
\end{eqnarray*}
We have used the Cauchy-Schwarz inequality and, for the last step, we have made the change of variables $\y = (x- \tilde x)/ \sqrt{2t}$. This proves that 
\[
\nr{\pppg x^{-s} \partial^\b e^{t\Dx} \pppg x^{-s}}_{\Lc(L^1(\R^d),L^2(\R^d))} \lesssim t^{- \frac d 4 - \frac {\abs{\b}} 2 - \frac s 2}.
\]
We similarly get 
\begin{eqnarray*}
\nr{\pppg x^{-s} \partial^\b e^{t\Dx} \pppg x^{-s}}_{\Lc(L^2(\R^d),L^2(\R^d))} &\lesssim& t^{ - \frac {\abs{\b}} 2 - \frac s 2},\\
\nr{\pppg x^{-s} \partial^\b e^{t\Dx} \pppg x^{-s}}_{\Lc(L^2(\R^d),L^\infty(\R^d))} &\lesssim& t^{- \frac d 4 - \frac {\abs{\b}} 2 - \frac s 2},\\
\nr{\pppg x^{-s} \partial^\b e^{t\Dx} \pppg x^{-s}}_{\Lc(L^1(\R^d),L^\infty(\R^d))} &\lesssim& t^{- \frac d 2 - \frac {\abs{\b}} 2 - \frac s 2}.
\end{eqnarray*}
The second estimate gives the case $s_1 = s_2 = 0$. We consider the case $s_1 = 0$, $s_2 = \frac d 2$. Since $\pppg x^{-\frac {\k d}2}$ defines a bounded operator from $L^2(\R^d)$ to $L^1(\R^d)$ we have 
\begin{equation*}
\nr{\pppg x^{-s} \partial^\b e^{t\Dx} \pppg x^{-s- \frac {\k d}2}}_{\Lc(L^2(\R^d),L^2(\R^d))}
\leq \nr{\pppg x^{-s} \partial^\b e^{t\Dx} \pppg x^{-s}}_{\Lc(L^1(\R^d),L^2(\R^d))} \lesssim t^{- \frac d 4 - \frac {\abs{\b}} 2 - \frac s 2}.
\end{equation*}
We similarly prove the cases $s_1 = \frac d 2$, $s_2 = 0$ and $s_1 = s_2 = \frac d 2$. Then the general case follows by interpolation, and the first estimate is proved. For the second estimate we write 
\[
\nr{\pppg x^{-\k s_1} \Dx e^{t\Dx} \pppg x^{-\k s_2}}_{\Lc(L^2(\R^d))} \leq \nr{\pppg x^{-\k s_1} e^{\frac {t\Dx}3}} \nr{\Dx e^{\frac {t\Dx} 3}} \nr{e^{\frac {t\Dx}3} \pppg x^{-\k s_2}}.
\]
On the right-hand side the middle factor is of size $O(t\inv)$ by functional calculus. For the first and third factors we use the previous estimate. This concludes the proof of the proposition.
\end{proof}

\begin{remark} \label{rem-s}
In this proof we can see the role of the parameter $s$ which appears in Theorems \ref{th-energy-decay} and \ref{th-chaleur}. When considering a spatial derivative of the solution, a factor of size $\abs{x-\tilde x}/(2t)$ appears in the kernel. We have to control this power of $x$. For this we can use the negative power of $t$ and the exponential factor $\exp \big(-\abs{x -\tilde x}^2/(4t^2)\big)$. However, in suitable weighted spaces, this power of $x$ is controled by the weights and the negative power of $t$ gives a better rate of decay for the energy.
\end{remark}

\subsection{Comparison between the damped wave and the heat equations} \label{sec-decay-difference}

Before the proofs of Theorems \ref{th-energy-decay} and \ref{th-chaleur}, we prove the following lemma.

\begin{lemma} \label{lem-res-time}
Let $\Kc$ be a Hilbert space and let $I$ be an open bounded interval of $\R$. Let $\n \geq 0$, $\nu_0 > \n$ and $C \geq 0$.
Let $\f \in C_0^\infty(I,\Kc)$ and $\p \in C^\infty(I,\C)$. Assume that for $m \in \N$ with $m \leq \n_0 + 1$ and $\t \in I$ we have  
\[
\nr{\f^{(m)}(\t)}_\Kc \leq C \left( 1 + \abs \t^{\n_0 - 1 -m}\right),
\] 
\[
\abs{\p^{(m)}(\t)} \leq C \quad \text{and} \quad \abs{\p'(\t)} \geq \frac 1 C.
\]
Then there exists $c \geq 0$ which only depends on $I$, $\n$, $\n_0$ and $C$ such that for all $t \geq 0$ we have 
\[
\nr{\int_{I} e^{-it\p(\t)} \f(\t)\, d\t}_{\Kc} \leq c \pppg t^{-\n} \exp\left(t \sup_I \Im (\p) \right).
\]
\end{lemma}

Notice that for a fixed $\f$ we can replace $\n$ by any real number. The interest of this result is to give the uniform decay we can deduce from the uniform estimates on the derivatives of $\f$.
This kind of lemma was already used in \cite{boucletr14,royer-dld-energy-space,KhenissiRo}.

\begin{proof}
For $t > 0$ we set 
\[
\Phi(t) = \int_{I} e^{-it\p(\t)} \f(\t)\, d\t.
\]
Let 
\[
\s = \sup_I \Im(\p).
\]

Since we can replace $\n_0$ by any $\tilde \n_0 \in ]\n,\n_0[$, we can assume without loss of generality that $\n_0$ is not an integer. We first remark that for all $t \geq 0$ we have 
\[
\nr{\Phi(t)}_\Kc \leq  \nr{\f}_{L^1(I,\Kc)} e^{\s t},
\]
so the difficulty comes from large values of $t$.
Let $m_0 \in \N$ be the integer part of $\n_0$ and $\g = m_0 + 1 - \n_0 \in ]0,1[$. By integrations by parts we obtain  
\begin{equation*}
(it)^{m_0} \Phi(t) = \int_{I} e^{-it\p(\t)} (L^{m_0}\f)(\t)\, d\t,
\end{equation*}
where 
\[
L = \frac d{d\t} \frac 1 {\p'(\t)}.
\]
Then we can write 
\begin{equation} \label{somme-Phi}
(it)^{m_0} \Phi(t) = \sum_{j=0}^{m_0} \int_{I} e^{-it\p(\t)} q_j(\t) \f^{(j)}(\t)\, d\t,
\end{equation}
where $q_j$ is a smooth function with bounded derivatives on $I$, and the bounds only depend on $C$. For $j \in \{0,\dots,m_0-1\}$ we can do another integration by parts. We have 
\begin{align*}
\nr{\int_{I} e^{-it\p(\t)} q_j(\t) \f^{(j)}(\t)\, d\t}_\Kc
& = \frac 1 {t} \nr{\int_I e^{-it\p(\t)} L (q_j \f^{(j)})(\t) \, d\t}_\Kc\\
& \lesssim \frac {e^{\s t}} t \int_I \left( \|\f^{(j)}\|_\Kc + \|\f^{(j+1)}\|_\Kc \right)\, d\t\\
& \lesssim \frac {e^{\s t}} t \int_I \abs \t^{-\g} \, d\t\\
& \lesssim \frac {e^{\s t}} t.
\end{align*}
Since $\f^{(m_0+1)}$ is not uniformly integrable near 0, we cannot proceed similarly for the last term of \eqref{somme-Phi}. We separate the contribution of $\t$ close to 0. Let $\b \in ]0,1[$. We have  
\[
\nr{\int_{\abs \t \leq t^{-\b}} e^{-it\p(\t)} q_{m_0}(\t) \f^{(m_0)}(\t)\, d\t}_\Kc \lesssim e^{\s t} \int_{\abs \t \leq t^{-\b}} \abs \t^{-\g} \, d\t \lesssim {e^{\s t}}{t^{-\b(1-\g)}}.
\]
For the contribution of $\t \in I$ with $\abs \t \geq t^{-\b}$ we can use an integration by parts as above.
\begin{multline*}
\nr{\int_{\abs \t \geq t^{-\b}} e^{-it\p(\t)} q_{m_0}(\t) \f^{(m_0)}(\t)\, d\t}_\Kc\\
\leq \frac {e^{\s t}} t \left(\nr{\f^{(m_0)}(t^{-\b})} + \nr{\f^{(m_0)}(-t^{-\b})} + \int_{\abs \t \geq t^{-\b}} \abs \t^{-\g - 1} \, d\t \right) \lesssim  {e^{\s t}t^{\b \g-1}}.
\end{multline*}
Finally we have 
\[
t^{m_0} \nr{\Phi(t)} \lesssim e^{\s t} t^{\b (\g - 1)},
\]
and the conclusion follows if we choose $\b$ close enough to 1.
\end{proof}

Now we can finish the proofs of Theorems \ref{th-energy-decay} and \ref{th-chaleur}:

\begin{proof}[Proof of Theorems \ref{th-energy-decay} and \ref{th-chaleur}]
We use the notation of Theorem \ref{th-chaleur}. 
Let $k > 1$. Let $U_0 = (u_0,iu_1) \in \Dom(\AcN^k)$ and $W_0 = (\AcN-i)^k U_0$. Let $u$ and $v$ be the solutions of \eqref{wave-neumann} and \eqref{heat}, respectively. For $t \geq 0$ we set
\[
U(t) = \begin{pmatrix} u(t) \\ i\partial_t u(t) \end{pmatrix} = e^{-it\AcN} U_0 \quad \text{and} \quad V(t) = \begin{pmatrix} v(t) \\ i\partial_t v(t) \end{pmatrix} = \begin{pmatrix} e^{t\Dx } P_0 (au_0+u_1) \\ i \Dx   e^{t\Dx } P_0 (au_0+u_1) \end{pmatrix}.
\]
We prove that there exists $C \geq 0$ such that for all $t \geq 1$ we have 
\begin{equation} \label{estim-UV1}
\nr{\nabla u(t) - \nabla v(t)}_{L^{2,-\d_1}(\O)} \leq C \left( t^{-\frac 12 (1 +s_1 +s_2 + \tilde s)} + \frac {\ln(t)^{k/2+1}} {t^{k/2}} \right) \nr{W_0}_{\HHN^{\d_2}}
\end{equation}
and
\begin{equation} \label{estim-UV2}
\nr{\partial_t u(t) - \partial_t v(t)}_{L^{2,-\d_1}(\O)} \leq C \left( t^{-\frac 12 (2 +s_1 +s_2 + \tilde s)} + \frac {\ln(t)^{k/2+1}} {t^{k/2}} \right) \nr{W_0}_{\HHN^{\d_2}}.
\end{equation}
Estimates \eqref{estim-th-uv} and \eqref{estim-th-duv} will follow. Moreover, with \eqref{estim-th-v} and \eqref{estim-th-dv} this will give Theorem \ref{th-energy-decay}.

\stepp 
Let $t \geq 1$. Given $\m \in ]0,1[$, we start from the identity
\[
e^{-it\AcN} (\AcN-i)^{-k} W_0  = \frac 1 {2i\pi} \int_{\Im(z) = \m} \frac {e^{-itz}}{(z -i)^k} (\AcN-z)\inv W_0 \, dz
\]
(see \cite{burq98}).%
By Proposition \ref{prop-inter-freq} there exists $\g > 0$ such that the result of Theorem \ref{th-high-freq} holds with $\t_0$ replaced by 1. 
Let $\vf \in C^0(\R)$ be such that $\vf(\t) = \m$ if $\t \in [-1,1]$, $\vf(\t) \in [-\g / 9, \m]$ if $\abs \t \in [1,2]$, $\vf(\t) = - \g /9$ if $\abs s \in [2,3]$ and $\vf(\t) = - \g / \t^2$ for $\abs \t \geq 3$. We can also assume that $\vf$ is smooth on $]-3,3[$. We set 
\[
\G = \singl{ \t + i \vf(\t),\t \in \R}.
\]
Then we have 
\[
e^{-it\AcN} (\AcN-i)^{-k} W_0  = \frac 1 {2i\pi} \int_{\G} \frac {e^{-itz}}{(z -i)^k} (\AcN-z)\inv W_0 \, dz.
\]
Let $\h_0 \in C_0^\infty(\R,[0,1])$ be supported in $]-3,3[$ and equal to 1 on $[-2,2]$, and $\h_1 = 1 - \h_0$. For $j \in \{0,1\}$ we set 
\[
I_j(t) =  \frac 1 {2i\pi} \int_{\G} \h_j(\Re(z)) \frac {e^{-itz}}{(z -i)^k} (\AcN-z)\inv W_0 \, dz.
\]

\stepp We begin with $I_0(t)$. We first notice that replacing $U_0$ by $W_0$ in order to introduce a factor $(z-i)^{-k}$ in the integral will be necessary to estimate the contribution $I_1(t)$ of high frequencies, but it is useless for low frequencies. It would not be a problem to obtain the uniform estimates of Theorem \ref{th-energy-decay} alone, but in order to prove the sharp result of Theorem \ref{th-chaleur} we have to go back to an expression in $U_0$ and without the factor $(z-i)^{-k}$. For this we use the identity 
\[
(\AcN-z)\inv (\AcN-i)^k = \sum_{j=1}^{k} (\AcN-i)^{k-j} (z-i)^{j-1} + (z-i)^k (\AcN-z)\inv.
\]
We set 
\[
\tilde \Th(z) = \begin{pmatrix} \tilde \Th_1(z) \\ \tilde \Th_2(z) \end{pmatrix} = \sum_{j=1}^{k} (\AcN-i)^{k-j} (z-i)^{j-1 - k} U_0 + \Th(z) U_0,
\]
where $\Th(z)$ is given by \eqref{def-Th}. Then by Theorem \ref{th-low-freq} we have $I_0(t) = I_\Heat (t) + I_\Th(t)$ where we have set
\[
I_\Heat(t) =  \frac 1 {2i\pi} \int_{\G}  e^{-itz} \h_0(\Re(z)) \RcC U_0 \, dz
\]
and
\[
I_\Th(t) = \begin{pmatrix} I_{\Th,1}(t) \\ I_{\Th,2}(t) \end{pmatrix} = \frac 1 {2i\pi} \int_{\G} e^{-itz} \h_0(\Re(z)) \tilde \Th(z) \, dz. 
\]

\stepp Let $\mathfrak s \in ]\tilde s, \min(2,d,\rho)[$. For $I_\Th$ we apply Theorem \ref{th-low-freq} (with $\mathfrak s$ instead of $s$) and Lemma \ref{lem-res-time}. This gives
\begin{equation} \label{estim-ITh1}
\nr{\nabla I_{\Th,1}(t)}_{L^{2,-\d_1}}  \lesssim e^{t\m}  t^{-\frac 12 (s_1 + s_2 + \tilde s + 1)} \nr{W_0}_{\HHN^{\d_2}}
\end{equation}
and
\begin{equation} \label{estim-ITh2}
\nr{I_{\Th,2}(t)}_{L^{2,-\d_1}}  \lesssim e^{t\m}  t^{-\frac 12 (s_1 + s_2 + \tilde s + 2)} \nr{W_0}_{\HHN^{\d_2}}.
\end{equation}

\stepp We now turn to $I_\Heat(t)$. We set
\[
I_\Heat(t) = 
\begin{pmatrix}
I_\Heat^1 (t) & I_\Heat^2 (t) \\
I_\Heat^3 (t) & I_\Heat^4 (t)
\end{pmatrix}.
\]
We consider the upper left coefficient $I_\Heat^1 (t) = i\Rcz P_0 a$. We denote by $E_N$ the spectral measure associated to $-\DN $ and $\Pi_N = E_N([0,\g/18])$. Then 
\begin{align*}
I_\Heat^1(t) u_0 =  i \int_{\Xi = 0}^{+\infty} \frac {1}{2i\pi} \int_{\G} \h_0(\Re(z)) \frac {e^{-itz}} {\Xi - i z}  \, dz \, dE_N(\Xi) P_0 a u_0.
\end{align*}
The integrand is holomorphic in $\big(]-2,2[ + i\R\big) \setminus \singl{-i\Xi}$ so by the residue theorem
\begin{align*}
\Pi_N I_\Heat^1(t) u_0
& = \int_{\Xi = 0}^{\frac \g {18}}   e^{-t\Xi}   \, dE_N(\Xi)P_0 a u_0  +  \frac {1}{2\pi} \int_{\Xi = 0}^{\frac \g {18}} \int_{\t = -3}^3 \h_0(\t) \frac {e^{-it(\t-i\g/9)}} {\Xi - i \t - \g /9}\, d\t \, dE_N(\Xi) P_0 a v_0  \\
& = \Pi_N e^{t\DN}  P_0 a u_0 +  \frac {1}{2\pi} \int_{\t = -3}^3 \h_0(\t) e^{-it\t- \g t/9} \Pi_N \big(-\DN  - i \t - \g/9\big)\inv P_0 a u_0 \, d\t.
\end{align*}
The second term can be considered as a rest:
\[
\nr{\frac {1}{2\pi} \int_{\t = -3}^3 \h_0(\t) e^{-it\t- \g t/9} \Pi_N \big(-\DN  - i \t - \g/9\big)\inv P_0 a u_0 \, d\t}_{L^2(\O)} \lesssim e^{-\frac {\g t} 9} \nr{u_0}_{L^2(\O)}.
\]
We also have $\nr{(1-\Pi_N) e^{t\DN} P_0 a u_0} \lesssim e^{- \frac {\g t} {18}} \nr{u_0}$ and after a change of contour (such that $\Im(z) = - \frac {\g}{36}$ if $\abs {\Re(z)} \leq 1$)
\[
\nr{(1-\Pi_N) I_\Heat^1(t) u_0} \lesssim e^{-\frac {\g t}{36}} \nr{u_0}.
\]
This yields
\[
\nr{\nabla  \big( I_\Heat^{1}(t)u_0 - e^{t\DN } P_0 a u_0\big)}_{L^2(\O)} \lesssim e^{-\frac {\g t} {36}} \nr{u_0}_{L^2(\O)}.
\]
We estimate similarly $iI_\Heat^2(t) u_1$, $I_\Heat^3(t) u_0$ and $i I_\Heat^4(t)u_1$ and finally obtain 
\begin{equation} \label{estim-Iheat}
\nr{I_\Heat(t) - V(t)}_{\EEN} \lesssim  e^{-\frac {\g t}{36}} \nr{U_0}_{\HHN}.
\end{equation}

\stepp For $I_1(t)$ we use the strategy of \cite{lebeau96,Burq-Hi-07}. We set
\[
\rho(t) = \left( \frac t {\ln(t)} \right)^{\frac 12}.
\]
Let $m > \frac d 2$, $c_2 \in \big]0, \sqrt {\g /(m+2)} \big[$ and $c_1 \in ]0,c_2[$. We can write 
\[
I_1(t) = \frac 1 {2i\pi} \frac 1 {\sqrt {2\pi}} \int_{z \in \G} \int_{\sigma \in \R} e^{-(z-\sigma)^2/2} \h_1(\Re(z)) \frac {e^{-itz}}{(z -i)^k} (\AcN-z)\inv W_0 \, d\sigma \, dz.
\]
Then we split $I_1(t)$ as $I_{11}(t) + I_{12}(t)$ where $I_{11}(t)$ and $I_{12}(t)$ are defined as $I_{1}(t)$ except that in $I_{11}(t)$ the integral on $\sigma$ is restricted to $[-c_1 \rho(t), c_1 \rho(t)]$ (and to its complement for $I_{12}(t)$).

\stepp We first consider $I_{11}(t)$. If $z \in \G$ is such that $\abs{\Re(z)} \leq c_2 \rho(t)$ then 
\[
\Im(z) \leq - \frac \g {c_2^2 \rho(t)^2},
\]
so 
\[
\nr{\frac {e^{-itz}}{(z -i)^k} (\AcN-z)\inv}_{\EEN} \lesssim e^{- \frac {t\g}{c_2^2 \rho(t)^2}}  \rho(t)^2  \lesssim t^{1 - \frac \g {c_2^2}}.
\]
On the other hand there exists $\e_0 > 0$ such that if $\abs \s \leq c_1 \rho(t)$ and $\abs{\Re(z)} \geq c_2 \rho(t)$ (with $t$ large enough) then we have 
\[
\Re \big( (\sigma-z)^2 \big) \geq 2 \e_0 (\rho(t)^2 + \Re(z)^2),
\]
so
\begin{equation} \label{estim-I11}
\nr{I_{11}(t)}_{\EEN} \lesssim \rho(t) t^{1 - \frac \g {c_2^2}} \nr{W_0}_{\EEN} + \int_{\substack {z \in \G \\\abs{\Re(z)} \geq c_2 \rho(t)}} e^{-2\e_0(\rho(t) + \Re(z)^2)} \nr{W_0}_{\EEN} \, dz  \lesssim t^{-m} \nr{W_0}_{\EEN}.
\end{equation}

\stepp It remains to estimate $I_{12}(t)$. In the integrand the factor $\frac 1 {(z-i)^k}$ is small when $\abs{\Re(z)}$ is large, and this is what we will use to obtain the time decay. On the other hand we observe that the factor $e^{-itz}$ is small when $\Im(z) \ll -1$, while $(\AcN-z)\inv$ is small for $\Im(z) \gg 1$. We also have to keep in mind that the factor $e^{-(z-\sigma)^2/2}$ can become large if $\Im(z-\sigma)$ is large compared to $\Re(z-\sigma)$.
For $\th \in [0,t]$ we set 
\[
J_t(\th) =  \frac 1 {2i\pi} \frac 1 {\sqrt {2\pi}} \int_{z \in \G} \int_{\abs {\sigma} \geq c_1 \rho(t)} e^{-(z-\sigma)^2/2} \h_1(\Re(z)) \frac {e^{-i\th z}}{(z -i)^k} (\AcN-z)\inv W_0 \, d\sigma \, dz
\]
and 
\[
K_t(\th) =  \frac 1 {2i\pi} \frac 1 {\sqrt {2\pi}} \int_{z \in \G} \int_{\abs {\sigma} \geq c_1 \rho(t)} e^{-(z-\sigma)^2/2} \h_1(\Re(z)) \frac {e^{-i\th z}}{(z -i)^k} W_0 \, d\sigma \, dz.
\]
We have 
\[
(\partial_\th + i \AcN) J_t(\th) = i K_t(\th),
\]
so 
\begin{equation} \label{dec-I12}
I_{12}(t) = J_t(t) = J_t(0) + i \int_0^t e^{-i(t-\th)\AcN} K_t(\th) \, d\th.
\end{equation}
It remains to estimate separately $J_t(0)$ (for which we no longer have the factor $e^{-itz}$) and the integral of $K_t(\th)$ (for which we no longer have the resolvent of $\AcN$).

\stepp We begin with $J_t(0)$. For this we follow \cite{lebeau96}. We write 
\begin{equation} \label{Jt0}
J_t(0) = \frac 1 {2i\pi}  \frac 1 {\sqrt {2\pi}}  \int_{\abs{\sigma} \geq c_1 \rho(t)} S(\sigma) \, d\sigma,
\end{equation}
where for $\sigma \in \R$ we have set 
\begin{equation} \label{def-Ssigma}
S(\sigma) = \int_{z \in \G} \h_1(\Re(z)) \frac {e^{-(z-\sigma)^2 / 2} }{(z-i)^k} (\AcN-z)\inv W_0 \, dz.
\end{equation}
In \eqref{Jt0} we estimate the integral over $[c_1 \rho(t),+\infty[$. The integral over $]-\infty,-c_1\rho(t)]$ is analogous.
We set 
\[
\O_t = \singl {\sigma \in \C \st \arg \big( \sigma - c_1\rho(t)  \big) \in \left[0 , \frac \pi 8 \right]}.
\]
For $\sigma \in \O_t$ we define the contour $\G_\sigma = \G_- \cup \G_{0,\sigma} \cup \G_{+,\sigma}$ where $\G_- = \G \cap \singl{\Re(z) \leq 3}$, $\G_{0,\sigma}$ is the line segment joining the points $3 - i\g /9$ and $3+i\Im(\sigma) + i$ and $\G_{+,\sigma}$ is the half-line $\singl{\Re(z) \geq 3, \Im(z) = \Im(\sigma) + 1}$. In \eqref{def-Ssigma} we replace the contour $\G$ by $\G_\sigma$, and we denote by $S_-(\sigma)$, $S_0(\sigma)$ and $S_+(\sigma)$ the contributions of $\G_-$, $\G_{0,\sigma}$ and $\G_{+,\sigma}$, respectively. There exists $\e_0 > 0$ such that for $\sigma \in \O$ and $z \in \G_- \cup \G_{0,\sigma}$ we have 
\[
\abs{e^{-(z-\sigma)^2/2}} \leq e^{-\e_0(\rho(t)^2 + \abs \sigma^2 + \abs z^2)},
\]
so 
\begin{equation} \label{estim-S-}
\nr{S_-(\sigma)} + \nr{S_0(\sigma)} \lesssim e^{-\e_0(\rho(t)^2 + \abs \sigma^2)} \nr{W_0}.
\end{equation}
On the other hand
\begin{equation} \label{estim-S+}
\nr{S_+(\sigma)} \lesssim \int_{\y = 3}^{+\infty} \frac {e^{-((\y-\Re(\sigma))^2 - 1 )/2}}{(1+ \y)^k} \frac 1 {\Im(\sigma)+1} \nr{W_0} \, d\y \lesssim \frac 1 {(1+\Re(\sigma))^k} \frac 1 {\Im(\sigma) +1}\nr{W_0}.
\end{equation}
The function $S_+$ is holomorphic in $\C$. We consider the contour
\[
\Sigma_t = \singl{c_1 \rho(t) + \y e^{\frac {i\pi}8} , \y \geq 0}.
\]
By \eqref{estim-S+} we have 
\begin{equation*} 
\int_{\sigma \geq c_1 \rho(t)} S_+(\sigma) \, d\sigma  =  \int_{\sigma \in \Sigma_t} S_+(\sigma) \, d\sigma
\end{equation*}
and
\[
\nr{\int_{\sigma \in \Sigma_t} S_+(\sigma) \, d\sigma} \lesssim \int_{\y \geq 0} \frac 1 {(c_1\rho(t) + \y)^k} \frac 1 {1 + \y}\nr{W_0} \, d\y \lesssim \frac {\ln(\rho(t))}{\rho(t)^k} \nr{W_0}.
\]
With \eqref{Jt0} and \eqref{estim-S-} this yields
\begin{equation} \label{estim-Jto}
\nr{J_t(0)}_{\EEN} \lesssim \frac{\ln(t)^{k/2+1}}{t^{k/2}} \nr{W_0}_{\EEN}.
\end{equation}

\stepp For the integral of $K_t(\th)$ we proceed similarly, but on the other side of the real axis. We write 
\begin{equation*} 
K_t(\th) = \frac 1 {2i\pi}  \frac 1 {\sqrt {2\pi}}  \int_{\abs{\sigma} \geq c_1 \rho(t)} S_\th^*(\sigma) \, d\sigma,
\end{equation*}
where for $\sigma \in \R$ we have set 
\begin{equation*} 
S_\th^*(\sigma) = \int_{z \in \G} \h_1(\Re(z)) e^{-(z-\sigma)^2 / 2} \frac {e^{-i\th z}}{(z-i)^k} W_0 \, dz.
\end{equation*}
Again we only consider the integral for $\s$ over $[c_1 \rho(t),+\infty[$. We denote by $\O_t^*$ the image of $\O_t$ by complex conjugation.
For $\sigma \in \O_t^*$ we define the contour $\G_\sigma^* = \G_- \cup \G_{0,\sigma}^* \cup \G_{+,\sigma}^*$ where $\G_{0,\sigma}^*$ is the line segment joining the points $3 - i\g /9$ and $3 + i\Im(\sigma)-i$ and $\G_{+,\sigma}^*$ is the half-line $\singl{\Re(z) \geq 3, \Im(z) = \Im(\sigma)-1}$. We denote by $S_{-,\th}^*(\sigma)$, $S_{0,\th}^*(\sigma)$ and $S_{+,\th}^*(\sigma)$ the contributions of $\G_-$, $\G_{0,\sigma}^*$ and $\G_{+,\sigma}^*$ in $S_\th^*(\sigma)$, respectively. As above 
\begin{equation} \label{estim-S-star}
\nr{S_{-,\th}^*(\sigma)} + \nr{S_{0,\th}^*(\sigma)} \lesssim e^{-\e_0(\rho(t)^2 + \abs \sigma^2)} \nr{W_0}.
\end{equation}
On the other hand
\begin{equation*}
\nr{S_{+,\th}^*(\sigma)} \leq \int_{\y = 3}^{+\infty} \frac {e^{-((\Re(\sigma)-\y)^2 - 1 )/2}}{(1+ \y)^k} e^{-\th (1+\abs{\Im(\s)})} \nr{W_0} \, d\y \lesssim \frac 1 {(1+\Re(\sigma))^k} e^{-\th (1+\abs{\Im(\s)})}\nr{W_0}.
\end{equation*}
Using the complex conjugation $\Sigma_t^*$ of $\Sigma_t$ we obtain
\begin{align*}
\int_{\th = 0}^t \nr{\int_{\sigma \in \Sigma_t^*} S_{+,\th}^*(\sigma) \, d\sigma} d\th 
& \lesssim \int_{\y \geq 0} \frac 1 {(c_1\rho(t) + \y)^k} \int_0^{+\infty} e^{-\th (1+\sin(\pi/8) \y)} \nr{W_0} \, d\th  \, d\y\\
& \lesssim \int_{\y \geq 0} \frac 1 {(c_1\rho(t) + \y)^k} \frac 1 {1 + \sin(\pi/8) \y} \nr{W_0}  \, d\y \\
& \lesssim \frac {\ln(\rho(t))}{\rho(t)^k} \nr{W_0},
\end{align*}
hence
\begin{equation*} 
\nr{\int_0^t e^{-i(t-\th)\AcN} K_t(\th) \, d\th}_{\EEN} \lesssim \int_0^t \nr{K_t(\th)}_{\EEN} \, d\th \lesssim \frac{\ln(t)^{k/2+1}}{t^{k/2}} \nr{W_0}_{\EEN}.
\end{equation*}
With \eqref{dec-I12} and \eqref{estim-Jto} this gives 
\begin{equation} \label{estim-I12}
\nr{I_{12}(t)}_{\EEN} \lesssim \frac{\ln(t)^{k/2+1}}{t^{k/2}} \nr{W_0}_{\EEN}.
\end{equation}

\stepp With \eqref{estim-ITh1}, \eqref{estim-ITh2}, \eqref{estim-Iheat}, \eqref{estim-I11} and \eqref{estim-I12} we obtain \eqref{estim-UV1} and \eqref{estim-UV2} up to a factor $e^{t\m}$. But all these estimates are uniform in $\m \in ]0,1[$, so we can take the limit $\m \to 0$. This gives \eqref{estim-UV1} and \eqref{estim-UV2} and concludes the proofs of Theorems \ref{th-energy-decay} and \ref{th-chaleur}.
\end{proof}

\section{Related problems} \label{sec-dirichlet}

In this section we discuss several problems close to \eqref{wave-neumann}. More precisely we explain how our analysis provides results about the global and local energy decay for damped wave equation on $\R^d$, on the wave guide with Dirichlet boundary conditions, and the damped Klein-Gordon equation.\\

\subsection{The damped wave equation in \texorpdfstring{$\R^d$}{Rd}}

We begin with Theorem \ref{th-Rd} about the energy decay on the Euclidean space. The analysis of low frequencies is the same as for the wave guide. Indeed we saw that the main contribution was given by functions which do not depend on $y \in \o$ but only on $x \in \R^d$. There is nothing else on $\R^d$. The main difference with the wave guide is that on $\R^d$ any classical trajectory goes to infinity and therefore meet the damping. Thus by Proposition \ref{prop-estim-res-Rd} we obtain that the corresponding resolvent $(\Ac-\t)\inv$ (where $\Ac$ is defined as $\AcN$ but on $\EE = \dot H^1(\R^d) \times L^2(\R^d)$ and there is no boundary so no boundary condition) is uniformly bounded on $\EE$ for $\abs \t \geq 1$. In this case we know that the contribution of high frequencies decays exponentially and without loss of derivative (see for instance \cite{royer-diss-schrodinger-guide} for a proof of the energy decay with a uniform bound on the resolvent in a slightly different setting). Thus the energy decay in 
this case is similar to the results 
on the wave guide except that the energy decay is only limited by the contribution of low frequencies. This is what was stated in Theorem \ref{th-Rd}.

\subsection{The damped wave equation on the wave guide with Dirichlet boundary condition}

Now we discuss the damped wave equation on the wave guide $\O$ with Dirichlet boundary condition:
\begin{equation} \label{wave-dirichlet} 
\begin{cases}
\partial_t^2 u  -\D u + a(x) \partial_t u = 0  & \text{on  }  \R_+ \times \O, \\
u = 0 & \text{on } \R_+ \times \partial \O,\\
\restr{(u , \partial_t u )}{t = 0} = (u_0, u_1) &  \text{on } \O.
\end{cases}
\end{equation}

Here we can replace the assumption \eqref{hyp-amort-inf} on the absorption index by the following weaker version: $a$ is bounded and there exist a compact subset $K$ of $\bar \O$ and $c_0 > 0$ such that 
\begin{equation} \label{hyp-amort-Dir}
\forall (x,y) \in \O \setminus K, \quad a(x,y) \geq c_0.
\end{equation}

Contrary to the Euclidean case, we will have the same problems as for \eqref{wave-neumann} for the contribution of high frequencies, but there is no longer any restriction due to low frequencies (which is why the absorption index no longer has to go to a constant at infinity). 

The operator $\AcD$ is defined by 
\[
\AcD = \begin{pmatrix} 0 & I \\ -\D & -i a \end{pmatrix}
\]
on the Hilbert space 
\[
\EED = H^1_0(\O) \times L^2(\O)
\]
with domain
\begin{equation} \label{dom-Ac-D}
\Dom(\AcD) = \singl{(u,v) \in \EED \st \AcD (u,v) \in \EED}.
\end{equation}

\begin{proposition} \label{prop-res-reelle-Dir}
The operator $\AcD$ is maximal dissipative on $\EED$. Moreover any $\t \in \R$ belongs to the resolvent set of $\AcD$.
\end{proposition}

\begin{proof}
\stepp For the maximal dissipativeness we proceed as with Neumann boundary conditions (see also \cite{alouik02} where the similar statement is proved in an exterior domain).

\stepp Now let $\t > 0$ (the case $\t < 0$ is analogous). Since we have weakened the assumption on the absorption index, we have to adapt the proof of Proposition \ref{prop-RNt}. Let $\tilde a = a + c_0 \1_{K}$. Then $\tilde a \geq c_0$ everywhere on $\O$. This implies that the operator $-\D - i\t( \tilde a - c_0)$ is maximal dissipative, and hence its spectrum is contained in the lower half-plane. Therefore the spectrum of $-\D -i\t \tilde a(x)$ (and in particular its essential spectrum) is a subset of $\singl{\Im(\z) \leq - \t c_0}$. Since $-\D -i\t a(x)$ is a bounded and relatively compact perturbation of $-\D -i\t \tilde a(x)$, we deduce by the Weyl Theorem (see Theorem \ref{th-weyl} in appendix, applied with $\Uc$ containing $\singl{\Im(z) > - \t c_0}$) that its essential spectrum is included in $\singl{\Im(\z) \leq - \t c_0}$. Then $\t^2$ belongs to the spectrum of $-\D -i\t a$ if and only it is an eigenvalue. As in the Neumann case we can check that this is not the case, that $(-\DDir - iza - z^2)\
\inv$ extends to a bounded operator in $\Lc(H^1(\O)',H^1(\O))$ uniformly in $z = \t + i\m$ with $\m \in [0,1]$, and consequently that 
$\t$ belongs to the resolvent set of $\AcD$.

\stepp It remains to prove that 0 also belongs to the resolvent set of $\AcD$. We use the separation of variables as in Section \ref{sec-separation-variables}. The only difference is that in this case the transverse operator is the Dirichlet Laplacian $\TD$ on $\o$, whose first eigenvalue is positive. We denote by $0 < \tilde \l_1 \leq \tilde \l_2 \leq \dots$ the eigenvalues for the transverse operator $\TD$ and by $(\tilde \f_k)_{k \in \N^*}$ a corresponding orthonormal sequence of eigenfunctions. Then as in Section \ref{sec-separation-variables} for $u = \sum_{k\in\N^*} u_k(x) \tilde \f_k (y) \in \Dom(-\DDir)$ we can write
\begin{align*}
\innp{-\DDir u}{u}_{L^2(\O)}
& = \sum_{k \in \N^*} \innp{(-\Dx + \tilde \l_k) u_k}{u_k}_{L^2(\R^d)}\\
& \geq \tilde \l_1 \sum_{k \in \N^*} \nr{u_k}_{L^2(\R^d)}^2\\
& \geq \tilde \l_1 \nr{u}_{L^2(\O)}^2.
\end{align*}
This proves that $-\DDir$ is invertible. Then, for $z \in \C$ small enough, this is then also the case for $(-\DDir - iza - z^2)$. We deduce that 0 is not in the spectrum of $\AcD$.
\end{proof}

With the same high frequency estimate as in Theorem \ref{th-high-freq} this proves that there exists $C \geq 0$ such that for all $t \in \R$ we have 
\begin{equation} \label{estim-res-dir}
\nr{(\AcD - \t)\inv}_{\Lc(\EED} \leq C \pppg \t^2.
\end{equation}
It only remains to apply the abstract result of \cite{BorichevTo10} to obtain the global energy decay for \eqref{wave-dirichlet}:

\begin{theorem} \label{th-energy-decay-dirichlet}
Let $k \in \N^*$. Then there exists $C \geq 0$ such that for all $U_0 \in \Dom(\AcD^k)$ and $t \geq 1$ we have 
\[
\nr{e^{-it\AcD} U_0}_{\EED} \leq C t^{-\frac k 2} \nr{\AcD^k U_0}_{\EED}.
\]
\end{theorem}

\subsection{Neumann boundary condition when \texorpdfstring{$a$}{a} is constant on each section}

We have seen that on the Euclidean space the energy decay is limited by the contribution of low frequencies, while on the wave guide with Dirichlet boundary conditions it is mainly a high frequency problem. The difficulty in Section \ref{sec-decay-difference} about the wave guide with Neumann boundary condition was that we had to deal with both difficulties at the same time.

In the analysis of $\RNz$ is Section \ref{sec-resolvent} we saw that the main contribution is given by $\RNz P_0$ for low frequencies (see \eqref{def-RcC}). On the other hand, for high frequencies the difficulty come from the contribution of high transverse frequencies. More precisely, $\RNz (1-P_0)$ extends to a holomorphic function around 0 (as for Dirichlet boundary condition), while $\RNz P_0$ behaves nicely for high frequencies (as in the Euclidean space).

For $U = (u,v) \in \EEN$ we set 
\begin{equation} \label{def-Pc0}
\Pc_0 U = \begin{pmatrix} P_0 u \\ P_0 v \end{pmatrix} \quad \text{and} \quad \Pc_0^\bot U = \begin{pmatrix} (1-P_0) u \\ (1-P_0) v \end{pmatrix}.
\end{equation}

Because of the factor $a$ in \eqref{expr-res}, we cannot simply say that $(\AcN-z)\inv \Pc_0$ behaves nicely for high frequencies and that $(\AcN-z)\inv \Pc_0^\bot$ behaves nicely for low frequencies. However, if $a$ only depends on $x$, then it commutes with $P_0$ and in this case we can reduce the Neumann problem to two simpler analyses.

\begin{proposition} \label{prop-low-freq-bot}
Assume that $a$ does not depend on $y \in \o$. By restriction, the operator $\AcN$ defines a maximal dissipative operator $\Acbot$ on $\EEbot = \Pc_0^\bot \EEN$ (with domain $\Dom(\Acbot) = \Dom(\AcN) \cap \EEbot$) such that $\s(\Acbot) \cap \R = \emptyset$.
\end{proposition}

\begin{proof}
Let $(u,v) \in \Dom(\AcN) \cap \EEbot$. We define the sequences $\seq u k$ and $\seq v k$ as in \eqref{def-uk}. By \eqref{sec-separation-variables} we have 
\[
-\D u - ia v = \sum_{k \geq 1} \big( - \Dx  u_k + \l_k u_k - i a v_k \big) \otimes \f_k \in \Ran(1-P_0),
\]
so
\[
\AcN \begin{pmatrix} u \\ v \end{pmatrix} = 
\begin{pmatrix}
v \\
- \D u -ia v
\end{pmatrix}
\in \EEbot.
\]
Similarly $(\AcN-z)\inv$ leaves $\EEbot$ invariant for any $z \in \C_+$, and the restriction is a bounded inverse for $(\Acbot -z)$. This proves that $\Acbot$ is maximal dissipative on $\EEbot$.
As for $\AcN$, any $\t \in \R \setminus \singl 0$ belongs to the resolvent set of $\Acbot$. And for $\t = 0$ we follow the same proof as for $\AcD$.
\end{proof}

Thus when the absorption index only depends on $x$ we can deal sperately with the projections on $\Pc_0\EEN$ and $\Pc_0^\bot \EEN$ of the solution. On $\Pc_0 \EEN$ everything is exactly as in the Euclidean space and on $\EEbot$ the resolvent $(\Acbot-z)\inv$ satisfies the same estimate as in \eqref{estim-res-dir}. In particular, as for Theorem \ref{th-energy-decay-dirichlet} we can use \cite{BorichevTo10} and hence we have no logarithmic loss as was the case in Theorem \ref{th-energy-decay}.

\begin{theorem} \label{th-energy-decay-bis} 
Assume that $a$ only depends on $x \in \R^d$.
\begin{enumerate} [(i)]
\item \label{th-item-bot}
Let $k \in \N^*$. Then there exists $C \geq 0$ such that for $t \geq 1$ and $U \in \Dom(\AcN^k)$ we have 
\[
\nr{e^{-it\AcN} \Pc_0^\bot U}_{\EEN} \leq C t^{-\frac k 2} \nr{\Pc_0^\bot U}_{\Dom(A^k)}.
\]
\item 
Let $s_1,s_2 \in \big[0,\frac d 2]$, $\k > 1$, $s \in [0,\min(d,\rho)[$ with $s \leq 1$, $\d_1 \geq  \k s_1 + s$ and $\d_2 \geq  \k s_2 + s$. Then there exists $C \geq 0$ such that for $t \geq 1$, $u_0,u_1 \in C_0^\infty(\bar \O)$ we have 
\[
\nr{e^{-it\AcN} \Pc_0 U}_{L^{2,-\d_1}(\O)^2} \leq C t^{-\frac 12 (1 + s_1 + s_2 + s)} \nr{\Pc_0 U}_{L^{2,\d_2}(\O)^2}.
\]
\end{enumerate}
\end{theorem}

\subsection{The damped Klein-Gordon equation}

\newcommand{\KG}{\mathrm{KG}}

We finally discuss the damped Klein-Gordon equation 
\[
\partial_t^2 u - \D u + m^2 u + a \partial_t u = 0,
\]
where $m > 0$, either on $\R^d$ or on $\O$ with Neumann or Dirichlet boundary conditions. The corresponding operator is now 
\[
\AcKG = \begin{pmatrix} 0 & I \\ -\D + m^2 & -i a \end{pmatrix},
\]
with suitable domain depending on the context. The corresponding resolvent on $L^2(\O)$ is then 
\[
R_{\KG}(z) = \big( - \D + m^2 - i z a - z^2 \big)\inv.
\]

If $a = 0$ we can check that the spectrum of $\AcKG$ is $\R \setminus ]-m,m[$. Moreover for $z$ close to $m$ or $-m$ the resolvent $R_{\KG}(z)$ behaves like $(-\D -\z)\inv$ for $\z$ close to 0. This explains why we have the same rate of decay for the energy as for the Schr\"odinger equation. With compactly supported damping, we have the same kind of behavior (see \cite{Malloug16}).

For the damped version of the Klein-Gordon equation with damping effective at infinity, the situation is quite different. Indeed, for the same reason as for the wave equation, any $\t \in \R \setminus \singl 0$ belongs to the resolvent set of $\AcKG$ (in particular there is no singularity for $\t = \pm m$). But, for the same reason as for the undamped Klein-Gordon equation, 0 also belongs to the resolvent set of $\AcKG$. Thus, there is no ``low frequency effect'' at all for the damped Klein-Gordon equation. On the other hand, for $\abs \t \gg 1$ the resolvent $R_{\mathrm{KG}}(\t)$ behaves as the wave analog ($m^2$ is negligible compared to $\t^2$) so, finally, the energy decay for the damped Klein-Gordon equation is the same as for the contribution of high frequencies for the corresponding wave equation. Thus, on the Euclidean space we have a uniformly bounded resolvent on all the real axis, hence a uniform exponential global energy decay for the time dependant problem (this is in fact a particular case of \
cite{BurqJo}). And in a wave guide with Neumann or Dirichlet boundary conditions we obtain the same decay as in Theorem \ref{th-energy-decay-dirichlet}.

\appendix 

\section{High frequency resolvent estimates on \texorpdfstring{$\R^d$}{Rd} with damping at infinity}  \label{sec-high-freq-Rd}

In this appendix we give two proofs for Proposition \ref{prop-estim-res-Rd}. This high frequency resolvent estimate is well known in weighted spaces when $a$ is compactly supported (or decays suitably at infinity). This can be proved either with semiclassical defect measures and the now usual contradiction argument (see \cite{gerardl93,lebeau96,burq02,jecko04}) or with the Mourre theory (see \cite{mourre83,amrein}) for dissipative operators (see \cite{royer-mourre,boussaidg10,boucletr14,royer-mourre-formes}). Both methods can be adapted in this situation. Notice that the setting of Proposition \ref{prop-estim-res-Rd} is quite simple and that both methods will prove to be efficient here. However it is interesting to have both of them for more general situations. For instance, the method with semiclassical measures allows a Schr\"odinger operator whose non-selfadjoint part has no sign (see \cite{royer-nondiss}) and/or is supported by the boundary of the domain (see \cite{royer-diss-wave-guide}). On the other 
hand the Mourre method can be applied to more general operators and requires less regularity.\\

For both proofs we rewrite the problem with semiclassical notation. With $h = \t \inv$ the estimate reads
\begin{equation} \label{estim-res-semiclass}
\nr{(-h^2 \Dx  - i h \a - 1)\inv}_{\Lc(L^2(\R^d))} \lesssim \frac 1 h, \qquad 0 \leq h \ll 1.
\end{equation}
We refer to \cite{zworski} for general results about semiclassical analysis.

\begin{proof}[Proof 1 (with semiclassical measures)]
Assume by contradiction that \eqref{estim-res-semiclass} is wrong. Then we can find sequences $\seq h m \in ]0,1]^\N$, $\seq \th m \in \R^\N$, $\seq u m \in H^2(\R^d)^\N$ and $\seq f m \in L^2(\R^d)^\N$ such that $h_m \to 0$, $\th_m \to 1$, $\nr{f_m}_{L^2(\R^d)} = o(h_m)$ and for all $m \in \N$ we have $\nr{u_m}_{L^2(\R^d)} = 1$ and 
\[
\big(-h_m^2 \Dx - ih_m \a - \th_m \big) u_m = f_m.
\]
The sequence $(u_m)$ is bounded in $L^2(\R^d)$ so, after extraction of a subsequence if necessary, there exists a Radon measure $\m$ on $\R^{2d}$ such that for all $q \in C_0^\infty(\R^{2d})$ we have 
\[
\innp{\Opwm (q) u_m}{u_m} \limt m \infty \int_{\R^{2d}} q \, d\m.
\]
We first observe that
\begin{equation} \label{estim-a-um}
\nr {\a u_m}_{L^2(\R^d)}^2 \lesssim \nr{\sqrt \a u_m}_{L^2(\R^d)}^2 = - \frac 1 {h_m} \Im \innp {(-h_m^2 \Dx - i h_m \a -\th_m)u_m}{u_m} \limt m \infty 0.
\end{equation}
Now let $q \in C^\infty(\R^{2d})$ be such that $q$ and all its derivatives are bounded, and $q(x,\x) = 0$ if $\abs{\abs \x^2 -1} \leq \frac 12$. Then for $m$ large enough we can define $\tilde q_m : (x,\x) \mapsto q(x,\x) / (\abs \x^2 - \th_m)$ and write 
\begin{align} \label{estim-tilde-q}
\lim_{m \to \infty} \innp{\Opwm(q) u_m}{u_m}
& = \lim_{m \to \infty} \innp{\Opwm(\tilde q_m) (-h_m^2 \D_x -\th_m) u_m}{u_m}\\
\nonumber
& = \lim_{m \to \infty} \innp{\Opwm(\tilde q_m) (-h_m^2 \D_x -ih_m \a - \th_m) u_m}{u_m} = 0.
\end{align}
Let $\h \in C_0^\infty(\R^d,[0,1])$ be such that $\a(x) \geq c_0$ on a neighborhood of $\supp(1-\h)$. \eqref{estim-a-um} and \eqref{estim-tilde-q} imply that $\m$ is supported in $\supp(\h) \times \singl{\big| \abs \x^2 -1 \big| \leq \frac 12}$. Moreover for $\tilde \h \in C_0^\infty(\R^d,[0,1])$ such that $\tilde \h(\x) = 1$ if $\big| \abs \x^2-1 \big| \leq \frac 12$ we have 
\[
\int_{\R^{2d}} \h(x) \tilde \h (\x) \, d\m (x,\x) = \lim_{m \to \infty} \innp{\Opwm(\h(x) \tilde \h(\x)) u_m}{u_m} = \lim_{m \to \infty} \innp{u_m}{u_m} = 1,
\]
so $\m \neq 0$. To obtain a contradiction, it remains to show that in fact $\m = 0$. For this we use the invariance of the support of $\m$ by the classical flow. Let $q \in C_0^\infty(\R^{2d},\R)$ be supported in $\singl{(x,\x) \st \big| \abs \x^2 - 1 \big| \leq \frac 34}$. We have 
\begin{align*}
\frac d {dt} \int_{\R^{2d}} q \circ \vf^t \, d\m = \int_{\R^{2d}} \big\{\abs \x^2 , q \circ \vf^t \big\}  \, d\m,
\end{align*}
where $\vf^t (x,\x) = (x+2t\x,\x)$ and $\{p,q\}$ is the Poisson bracket $\partial_\x p \cdot \partial_x q - \partial_x q \cdot \partial_\x q$. With \eqref{estim-a-um} we obtain that for any $\tilde q \in C_0^\infty(\R^{2d},\R)$ we have 
\begin{align*}
\int_{\R^{2d}} \big\{ \abs \x^2 , \tilde q \big\}  \, d\m
& = \lim_{m \to \infty} \frac i {h_m} \innp{[-h_m^2 \Dx , \Opwm(\tilde q)] u_m}{u_m}\\
& = - \lim_{m \to \infty} \frac 2 {h_m} \Im \innp{\Opwm(\tilde q) u_m}{(-h_m^2 \Dx - i h_m \a - \th_m)u_m}\\
& = 0.
\end{align*}
This proves that the integral of $q \circ \vf^t$ does not depend on $t$. But for some $t$ large enough the supports of $q \circ \vf^t$ and $\m$ are disjoint. This yields $\int q \, d\m = 0$. Then $\m = 0$, which gives the contradiction and concludes the proof of the proposition.
\end{proof}

Now we give a proof based on the Mourre's commutators method. For a perturbation of a Laplace operator, it is usual to consider (a perturbation of) the generator of dilations 
\[
\tilde A_h = - \frac {ih}2 (x \cdot \nabla + \nabla \cdot x) = \Opw(x\cdot \x)
\]
as conjugate operator. Here, due to the damping at infinity, we will only need a ``localized version'' of $\tilde A_h$. We refer to \cite{royer-mourre-formes} for the general theorem.

\begin{proof}[Proof 2 (with Mourre's method)]
Let $\h \in C_0^\infty(\R^d,[0,1])$ be such that $\a(x) \geq c_0$ on a neighborhood of $\supp(1-\h)$. Let $\tilde \h \in C_0^\infty(\R,[0,1])$ be supported in $\big]\frac 12, 2\big[$ and equal to 1 on $\big[\frac 34,\frac 32\big]$. For $h \in ]0,1]$ we set 
\[
A_h = \Opw \big( (x \cdot \x) \h(x) \tilde \h (\abs \x^2) \big).
\]
This operator is selfadjoint and bounded on $L^2(\R^d)$. We check that it satisfies all the assumptions of \cite[Theorem 4.1]{royer-mourre-formes}. Let $h \in ]0,1]$ and $t \geq 0$. For $u \in C_0^\infty(\R^d)$ we have 
\begin{align*}
\nr{\nabla e^{-i \th A_h}u - e^{-i \th A_h} \nabla u}
& \leq \int_0^\th \nr{\frac d {ds} e^{-is A_h} \nabla e^{-i(\th-s)A_h} u} ds\\
& \leq \int_0^\th \nr{e^{-is A_h} [A_h,\nabla] e^{-i(\th-s) A_h} u}ds\\
& \lesssim \nr{u}.
\end{align*}
This proves that the form domain $H^1(\R^d)$ is invariant by $e^{-i\th A_h}$. By pseudo-differential calculus the commutators 
\[
[-h^2 \Dx , A_h], \quad [-h^2 \Dx -ih\a,A_h] \quad \text{and} \quad [[-h^2 \Dx -ih\a,A_h] , A_h]
\]
extend to bounded operators on $L^2(\R^d)$ uniformly in $h \in ]0,1]$. It remains to check the main point, namely the lower bound of the commutator. There exists $\b > 0$ such that for all $(x,\x) \in \R^{2d}$ we have 
\[
\left\{ \abs \x^2 , (x \cdot \x) \h(x) \tilde \h(\abs \x^2) \right\} + \b \a (x) = 2 \abs \x^2  \h(x) \tilde \h(\abs \x^2) + 2 (x\cdot \x) \x \cdot \nabla \h(x) \tilde \h (\abs \x^2)+ \b \a(x) \geq  \tilde \h(\abs \x^2).
\]
After quantization and multiplication by $h$ we obtain
\[
[-h^2 \Dx , i A_h] + \b h \a(x) \geq  h \tilde \h(-h^2\Dx) + O(h), 
\]
where the rest is estimated in $\Lc(L^2(\R^d))$. We set $J = \big[\frac 34, \frac 32\big]$. We compose this inequality by the spectral projection $\1_J(-h^2\Dx)$ on both sides, and for $h$ small enough we obtain
\[
\1_J (-h^2\Dx) \big([-h^2 \Dx , i A_h] + \b h \a(x) \big) \1_J (-h^2\Dx) \geq  \frac h 2 \1_J (-h^2\Dx). 
\]
This is the Mourre assumption in the dissipative and semiclassical setting. By \cite{royer-mourre-formes} we obtain that there exists $h_0 \in ]0,1]$ and $c \geq 0$ such that for $h \in ]0,h_0]$ we have 
\[
\nr{\pppg {A_h} \inv (-h^2\Dx -ih\a - 1)\inv \pppg {A_h}\inv}_{\Lc(L^2(\R^d))} \leq \frac c h.
\]
But $\pppg {A_h}$ is a bounded operator on $L^2(\R^d)$, so we easily deduce the same estimate without weight. This gives \eqref{estim-res-semiclass} and concludes the proof.
\end{proof}

Notice that we did not make standard use of the Mourre theory. Indeed we did not have to prove the limiting absorption principle in some weighted space, since here the resolvent is well defined even on the real axis. The point was only to use the parameter dependant version of the abstract result to obtain uniform estimate for this resolvent.

\section{Weyl's essential spectrum Theorem} \label{sec-weyl}

\newcommand{\disc}{{\mathrm{disc}}}
\newcommand{\ess}{{\mathrm{ess}}}

In this section we briefly discuss the essential spectrum of a non-selfadjoint operator. We first recall that there are different reasonable definitions which coincide for selfadjoint operators but not in the general case (see for instance \cite{Schechter66,GustafsonWei69,EdmundsEva}). 

Here we follow \cite{rs4}. Let $A$ be a closed operator on some Hilbert space $\Kc$. We denote by $\rho(A)$ the resolvent set of $A$. Let $\l$ in the spectrum $\s(A)$ of $A$. We say that $\l$ is in the discrete spectrum $\s_\disc(A)$ of $A$ if it is isolated in $\s(A)$ and if the projection
\[
\frac 1 {2i\pi} \int_{\abs{\s-\l} = r} (A-\s)\inv \, d\s
\]
(where $r > 0$ is such that $\s(A) \cap D(\l,2r)=\singl{\l}$) is of finite rank. Then we define the essential spectrum of $A$ by $\s_\ess(A) = \s(A) \setminus \s_\disc(A)$. It is a closed subset of $\C$.

With this definition, the essential spectrum of $A$ is preserved by perturbation by a relatively compact operator if $A$ is selfadjoint but not in the general case (see Corollary 2 and Example 1 in Section XIII.4 of \cite{rs4}). It may happen that a connected component of $\C \setminus \s_\ess(A)$ is included in the spectrum (and therefore in the essential spectrum) of the perturbed operator. However, we can check that this is the only problem which can occur. Thus, if we can show that some connected component of $\C \setminus \s_\ess(A)$ intersects the resolvent set of the perturbed operator, it has empty intersection with its essential spectrum:

\begin{theorem}[Weyl's essential spectrum theorem] \label{th-weyl}
Let $A$ and $B$ be closed operators such that $(B-A)$ is $A$-compact. Let $\Uc$ be a connected component of $\C \setminus \s_\ess(A)$. Then 
\[
\Uc \cap \rho(A) \cap \rho(B) \neq \emptyset \quad \implies \quad \Uc \cap \s_\ess(B) = \emptyset.
\]
\end{theorem}

The proof follows the same lines as the proof of Theorem XIII.14 in \cite{rs4}. We recall the ideas.

\begin{proof}
Let $z_0 \in \Uc \cap \rho(A) \cap \rho(B)$, $R_A = (A-z_0)\inv$ and $R_B = (B-z_0)\inv$. By the resolvent identity the difference $D := R_A - R_B = R_B (B-A) R_A$ is compact. Let $\Vc$ be the (connected and open) set of $z \in \C^*$ such that $z_0 + z\inv \in \Uc$. By Lemma 2 in \cite[Section XIII.4]{rs4} we have $\Vc \cap \rho(R_A) \cap \rho(R_B) \neq \emptyset$ and it is enough to prove that $\Vc \subset \s_\ess(R_B)$.

Let $z \in \Vc$. If $z \in \rho(R_A)$ then $(R_B-z)\inv$ exists if and only if $(1 - D (R_A-z)\inv)\inv$ exists. This is the case for at least one point in $\Vc$. Moreover the map $z \mapsto 1 - D (R_A-z)\inv$ is meromorphic in $\Vc$. By the meromorphic Fredholm Theorem (Theorem XIII.13 in \cite{rs4}) we obtain that $(1 - D (R_A-z)\inv)\inv$ exists on $\Vc$ except for a discrete set where it has finite rank residues. Thus $R_B$ only has discrete spectrum in $\Vc$, and the conclusion follows.
\end{proof}

\end{document}